\documentclass[journal]{IEEEtran}
\ifCLASSINFOpdf
\else
\fi
\usepackage{lipsum}
\usepackage{amsmath}
\usepackage{math-common}
\usepackage{xcolor}
\usepackage{graphicx}
\usepackage{amsthm}
\usepackage{algorithm}
\usepackage{algpseudocode}%
\usepackage{booktabs}
\usepackage{multirow}
\usepackage[noadjust]{cite}

\graphicspath{ {./img/} }

\newtheorem{theorem}{Theorem}
\newtheorem{lemma}{Lemma}
\newtheorem{definition}{Definition}
\newtheorem{corollary}{Corollary}
\newtheorem{prop}{Proposition}

\begin{document}
%
\title{Covariance-Free Sparse Bayesian Learning}
%
%
%

\author{Alexander~Lin,~\IEEEmembership{Student Member,~IEEE,}
        Andrew~H.~Song,~\IEEEmembership{Student~Member,~IEEE,} \\
        Berkin~Bilgic,
        and~Demba~Ba,~\IEEEmembership{Member,~IEEE}
\thanks{A. Lin and D. Ba are with the School of Engineering and Applied Sciences, Harvard University, Cambridge,
MA, 02138 USA (e-mail: alin@seas.harvard.edu; demba@seas.harvard.edu).}
\thanks{ A. H. Song is with the Electrical Engineering and Computer Science, Massachusetts Institute of Technology, Cambridge, MA, 02138 USA (email: andrew90@mit.edu).}
\thanks{{B. Bilgic is with Harvard-MIT Health Sciences and Technology, Massachusetts Institute of Technology, Cambridge, MA, USA, 
Athinoula A. Martinos Center for Biomedical Imaging, Charlestown, MA, USA 
Department of Radiology, Harvard Medical School, Boston, MA, USA.}}}
\maketitle

\begin{abstract}
Sparse Bayesian learning (SBL) is a powerful framework for tackling the sparse coding problem while also providing uncertainty quantification.  The most popular inference algorithms for SBL exhibit prohibitively large computational costs for high-dimensional problems due to the need to maintain a large covariance matrix. To resolve this issue, we introduce a new method for accelerating SBL inference -- named covariance-free expectation maximization (CoFEM) -- that avoids explicit computation of the covariance matrix.  CoFEM solves multiple linear systems to obtain unbiased estimates of the posterior statistics needed by SBL.  This is accomplished by exploiting innovations from numerical linear algebra such as preconditioned conjugate gradient and a little-known diagonal estimation rule.
For a large class of compressed sensing matrices, we provide theoretical justifications for why our method
scales well in high-dimensional settings.
Through simulations, we show that CoFEM can be up to thousands of times faster than existing baselines without sacrificing coding accuracy. Through applications to calcium imaging deconvolution and multi-contrast MRI reconstruction, we show that CoFEM enables SBL to tractably tackle high-dimensional sparse coding problems of practical interest. 
\end{abstract}

\IEEEpeerreviewmaketitle

\section{Introduction}
%
%
%
%
\IEEEPARstart{S}{parse} coding is a fundamental problem in signal processing that seeks a sparse solution $\bd z$ to the equation
\begin{align}
\bd y = \bd \Phi \bd z + \bd \varepsilon, \label{sparse-coding}
\end{align}
where $\bd z \in \R^D$ is a latent sparse vector, $\bd y \in \R^N$ is an observed measurement vector, $\bd \Phi \in \R^{N \times D}$ is a known dictionary, and $\bd \varepsilon \in \R^N$ is an unknown noise vector. 
 
 \emph{Sparse Bayesian learning} (SBL) is an effective methodology for sparse coding.  It has been employed in several popular models, such as sparse Bayesian regression \cite{mackay1996bayesian}, relevance vector machines \cite{tipping2001sparse}, and Bayesian compressed sensing \cite{ji2008bayesian, ji2008multitask, wipf2004sparse}.  The practical applications  of SBL are numerous, encompassing diverse examples such as medical image reconstruction \cite{bilgic2011multi, liu2018image}, direction of arrival estimation \cite{ chen2018off}, human pose estimation \cite{ ma2006sparse}, structural health monitoring \cite{zhang2012compressed}, seismic exploration \cite{yuan2017sparse}, and visual tracking \cite{williams2005sparse}. 

SBL offers several advantages compared to other common approaches to sparse coding (e.g. {$\ell_0$ or $\ell_1$ regularization}).  As a Bayesian method, SBL provides uncertainty quantification and the ability to recover credible intervals over $\bd z$ instead of a single point solution.   Moreover, SBL does not need to tune regularization parameters since it can learn them or integrate them out using hyperpriors \cite{ji2008multitask}.  As a generative model, SBL can be embedded as a submodule within a larger framework to enforce more complex structure for $\bd z$ (e.g. group sparsity \cite{wu2015space}, block sparsity \cite{ fang2014pattern}).  Finally, SBL has favorable optimization properties, such as a sparser global minimum than $\ell_1$ methods and fewer local minima than $\ell_0$ methods \cite{wipf2007new}.

One often-noted limitation of SBL is the heavy computational cost of its inference algorithm \cite{tipping2001sparse, bilgic2011multi}.  On the one hand, the fact that SBL requires more computation than non-Bayesian approaches should be unsurprising, since SBL recovers an entire distribution instead of a single point estimate.  On the other hand, the most widely-used options for SBL inference scale poorly to high-dimensional problems, which are becoming increasingly common in many domains.  This limitation threatens to render the SBL paradigm obsolete for large-scale settings, as inference cannot be performed in an acceptable timeframe for practical applications.  

The principal inference procedure of SBL is the expectation-maximization (EM) algorithm \cite{mackay1996bayesian, tipping2001sparse}.  Each iteration of EM is computationally expensive, requiring $O(D^3)$-time and $O(D^2)$-space to invert a large covariance matrix, where $D$ is the dimension of the sparse codes. There have been several attempts to reduce the costs of EM, using iteratively reweighted least-squares (IRLS) \cite{wipf2010iterative}, approximate message passing (AMP) \cite{fang2016two}, and variational inference (VI) \cite{duan2017fast}. A popular alternative to EM called the sequential algorithm (Seq) \cite{tipping2003fast} is often faster in practice. However, these methods lack either \emph{scalability} at very high dimensions $D$ or \emph{accuracy} in recovering ground-truth sparse codes when compared to EM.


\subsection{Contributions}
We introduce a novel method for accelerating SBL's EM algorithm for high-dimensional problems.  We call our method \textbf{covariance-free expectation maximization (CoFEM)} because it eliminates the main bottleneck of EM -- i.e. the storage and inversion of the covariance matrix.  We demonstrate that CoFEM has significant advantages in both scalability and accuracy over other SBL approaches (reviewed in Section \ref{sec:model}).  Our contributions are categorized into four types -- methodology (Section \ref{sec:cofem}), theory (Section \ref{sec:theory}), experimentation (Section \ref{sec:simulated}), and applications (Section \ref{sec:experiments}).

\emph{Methodology:} 
We develop CoFEM, which solves multiple linear systems and exploits a little-known diagonal estimation rule to obtain \emph{unbiased} estimates of the posterior moments needed by EM.  The multiple systems can be solved in parallel using an iterative solver, such as the conjugate gradient (CG) algorithm, without constructing the covariance matrix.  This simple, yet powerful principle reduces EM's per-iteration time complexity from $O(D^3)$ to $O(UK\tau_D)$, where $\tau_D$ is the time needed for matrix-vector multiplication and $U, K \ll D$ are hyperparameters for the number of CG steps and the number of linear systems.  
Furthermore, CoFEM reduces EM's space complexity from $O(D^2)$ to $O(DK)$. 

\emph{Theory:} We prove new theorems that further characterize CoFEM's asymptotic complexities.  In applying this theory to matrices $\bd \Phi$ that satisfy the restricted isometry property (RIP) and commonly occur in compressed sensing applications, we show that, in the limit of a large number of EM iterations, CoFEM's hyperparameters $U$ and $K$ can be kept small even as the dimensionality of the problem $D$ grows very large.  We use this theory to justify practical choices for $U$ and $K$ that need not increase with $D$, leading to signficant computational savings compared to existing SBL algorithms.

\emph{Experimentation:} We perform simulations comparing CoFEM to several existing SBL algorithms \cite{tipping2001sparse}, \cite{wipf2010iterative}, \cite{fang2016two}, \cite{duan2017fast}, \cite{tipping2003fast}, for high-dimensional sparse signal recovery.  Across all of our settings, CoFEM is able to maintain the same level of accuracy as EM due to its unbiased estimation of posterior moments.
In addition, CoFEM's highly-parallelized and space-saving design enables further acceleration via graphics processing units (GPUs); most other approaches cannot fully realize this benefit due to their heavy memory requirements.  In practice, CoFEM with GPU acceleration can reduce hours of computation for covariance-based algorithms to a few seconds.


\emph{Applications:} We apply CoFEM-equipped SBL in two settings of practical interest: (1) calcium deconvolution and (2) multi-contrast MRI reconstruction.  In these high-dimensional settings, CoFEM attains competitive computation time with non-Bayesian approaches, while providing several benefits (e.g. superior performance, uncertainty quantification).  These applications require extensions of the SBL model to multi-task learning and to settings with non-negativity constraints, which we demonstrate CoFEM is flexible enough to handle.       

\section{Sparse Bayesian Learning} \label{sec:model}

\subsection{Generative Model}
To solve Eq. \eqref{sparse-coding}, SBL imposes the following model:
\begin{align}
\bd z &\sim \mathcal{N}(\bd 0, \mathrm{diag}\{\bd \alpha\}^{-1}), \nonumber \\
\bd y \given \bd z &\sim \mathcal{N}(\bd \Phi \bd z, 1 / \beta \mathbf I), \label{model}
\end{align}
where $\beta \in \R$ is the inverse of the variance of the noise (commonly called precision) and $\mathbf I$ is the identity matrix.  Given $\bd y$, SBL performs inference on this model to recover $\bd z$.

The main identifying feature of SBL is the diagonal Gaussian prior with precision parameters $\bd \alpha \in \R^D$ for $\bd z$.  The notation $\mathrm{diag}\{\bd \alpha\}$ in Eq. \eqref{model} maps $\bd \alpha$ to a $D \times D$ matrix with $\bd \alpha$ along its diagonal and zero elsewhere.  SBL performs type II maximum likelihood  estimation \cite{wipf2011latent} by integrating out $\bd z$ and optimizing $\bd \alpha$.  Thus, SBL learns a posterior distribution with uncertainty over $\bd z$.  The overall learning objective is:                       
\begin{align}
\max_{\bd \alpha}& \log p(\bd y \given \bd \alpha) = \log \int_{\bd z} p(\bd y \given \bd z) p(\bd z \given \bd \alpha) d \bd z  \label{mle2}.
\end{align}

As this objective is optimized, many of the elements of $\bd \alpha$ diverge to $\infty$.  For these elements, the independent Gaussian priors converge to point masses at zero, forcing their respective posteriors to follow suit.  Thus, upon convergence of $\bd \alpha$ {to $\bhat \alpha$}, the recovered posterior distribution {$p(\bd z \given \bd y, \bhat \alpha)$} is often highly sparse.  This phenomenon is called \emph{automatic relevance determination} \cite{wipf2007new} because SBL learns which elements of $\bd z$ are ``relevant" (i.e. non-zero) from the data.

In the remainder of this section, we review several inference schemes that have been proposed to optimize Eq. \eqref{mle2} and comment on their respective shortcomings.    

\subsection{Expectation-Maximization Algorithm} \label{sec:inference}
Most SBL inference schemes are built on the \emph{expectation-maximization} (EM) algorithm, {a general framework for parameter estimation in the presence of latent variables \cite{dempster1977maximum, tipping2001sparse}.  EM iteratively alternates between an expectation step (E-Step) and a maximization step (M-Step). Let $\bidx \alpha t$ be the solution at the start of the $t$-th iteration.  The E-Step computes the expectation $Q(\bd \alpha; \bidx \alpha t)$ of the complete data log-likelihood $\log p(\bd z, \bd y \given \bd \alpha)$ with respect to the latent posterior $p(\bd z \given \bd y, \bidx \alpha t)$:}  
\begin{align}
&Q(\bd \alpha; \bidx \alpha t) = \E_{p(\bd z \given \bd y, \bidx \alpha t)}[\log p(\bd z, \bd y \given \bd \alpha)] \nonumber \\
&= \E_{p(\bd z \given \bd y, \bidx \alpha t)}[\log p(\bd z \given \bd \alpha) + \log p(\bd y \given \bd z)] \nonumber \\
&\propto \sum_{j=1}^D \log \alpha_j - \alpha_j  \cdot \E_{p(\bd z \given \bd y, \bidx \alpha t)}\left[z_j^2\right] + \text{const},
 \label{estep}
\end{align}
where ``const" absorbs all terms that are constant with respect to $\bd \alpha$.  The posterior $p(\bd z \given \bd y, \bidx \alpha t)$ is a multivariate Gaussian distribution $\mathcal{N}(\bd \mu, \bd \Sigma)$ with mean and covariance parameters
\begin{align}
\bd \mu = \beta \bd \Sigma \bd \Phi^\top \bd y, & & \bd \Sigma = (\beta \bd \Phi^\top \bd \Phi + \mathrm{diag}\{\bidx \alpha t\})^{-1}. \label{estep-params}
\end{align}  
The second moment of each $z_j$ in Eq. \eqref{estep} can be decomposed into a sum over the squared mean and variance, i.e.
\begin{align}
\E_{p(\bd z| \bd y, \bidx \alpha t)}[z_j^2] &= \E_{p(\bd z| \bd y, \bidx \alpha t)}[z_j]^2 + \Var_{p(\bd z| \bd y, \bidx \alpha t)}[z_j] \nonumber\\
&= \mu_j^2 + \Sigma_{j, j}, \label{second-moment}
\end{align} 

The M-Step updates each $\idx[j] \alpha t$ by maximizing Eq. \eqref{estep} with respect to $\alpha_j$.  Given $\bd \mu$ and $\bd \Sigma$, this can be done in closed-form by differentiating $Q$:
\begin{align}
\frac{\partial Q}{\partial \alpha_j} =  \frac{1}{\alpha_j} - {\mu_j^2 + \Sigma_{j, j}} = 0 \implies {\alpha}_j^{(t + 1)} = \frac{1}{\mu_j^2 + \Sigma_{j, j}}, \label{mstep}
\end{align} 
EM repeats Eqs. \eqref{estep-params} and \eqref{mstep} for $T$ iterations until convergence, while guaranteeing non-negative change in the log-likelihood objective of Eq. \eqref{mle2} at each step.

{Despite its simplicity, the EM algorithm is limited by its computational cost}.  The E-Step of Eq. \eqref{estep-params} is expensive {for large $D$}.  Storing $\bd \Sigma$ requires $O(D^2)$-space and computing it through matrix inversion requires $O(D^3)$-time.  This makes the standard EM algorithm challenging at high dimensions.

\subsection{Previous Attempts to Accelerate EM}
There have been several attempts to accelerate the EM algorithm for SBL inference.  One approach is based on \emph{iteratively reweighted least squares} (IRLS) \cite{wipf2010iterative}, which applies the Woodbury matrix identity to $\bd \Sigma$ in Eq. \eqref{estep-params}, yielding
\begin{align}
\bd \Sigma = \bold C - \bold C \bold \Phi^\top (\bold I + \bold \Phi \bold C \bold \Phi^\top)^{-1} \bold \Phi \bold C, \label{woodbury}
\end{align}
where $\bold C = \mathrm{diag}\{\bidx \alpha t\}^{-1}$.  By inverting an $N \times N$ matrix (as opposed to a $D \times D$ matrix), IRLS reduces the per-iteration time complexity of EM to $O(N^3 + N^2D)$. For problems in which $N \ll D$, IRLS is more efficient than EM.
However, if $N = O(D)$, the time complexity is still $O(D^3)$. 
Furthermore, like EM, IRLS requires storage of $\bd \Sigma$, which remains an expensive $O(D^2)$-space cost.        

Another line of work is based on variants of \emph{approximate message passing} (AMP).  AMP combines quadratic and Taylor series approximations with loopy belief propagation to iteratively estimate the means and variances of $\bd z$ in Eq. \eqref{second-moment}, circumventing matrix inversion \cite{fang2016two, al2014sparse}.  Though faster than EM in practice, AMP is known to diverge for dictionaries $\bd \Phi$ that do not satisfy zero-mean, sub-Gaussian criteria \cite{al2017gamp}.  

A third common strategy is to employ \emph{variational inference} (VI), which approximates the true posterior $p(\bd z \given \bd y, \bd \alpha)$ with a simpler surrogate $q(\bd z)$~\cite{bishop2000variational, babacan2011low, shutin2011fast}.  This allows for SBL inference that is inverse-free \cite{duan2017fast, worley2019scalable}.  However, VI approaches optimize a lower bound on Eq. \eqref{mle2} instead of the true log-likelihood objective.  Thus, they may converge to a sub-optimal solution for $\bd \alpha$.  Both AMP and VI are ultimately limited by the fact that their approximations to the means and variances of $\bd z$ can be \emph{biased} for general dictionaries $\bd \Phi$ \cite{al2017gamp, worley2019scalable}.  In Section \ref{sec:cofem}, we present a new method that ensures an \emph{unbiased} estimation of these moments, regardless of the structure of $\bd \Phi$.         
  
\subsection{Sequential Algorithm} \label{seq-alg}
The sequential (Seq) algorithm, is a notable alternative to EM that reduces computation time and space in practice \cite{tipping2003fast, ji2008bayesian}. Seq maintains a set $\mathcal{S} \subseteq \{1, 2, \ldots, D\}$  of ``active" indices such that $\alpha_j$ is finite for each $j \in \mathcal{S}$ and $\alpha_j = \infty$ for all $j \neq \mathcal{S}$.  Initially, $\mathcal{S} = \emptyset$.  {Indices are sequentially added to or deleted from $\mathcal{S}$ if making such a change can increase the log-likelihood objective (Eq. \eqref{mle2}).}  

At any given point, Seq only needs to store parts of the mean vector $\bd \mu$ and covariance matrix $\bd \Sigma$ corresponding to $\mathcal{S}$; all other components are assumed to be zero.  Thus, for truly sparse vectors $\bd z$ with $d$ non-zero components such that $d \ll D$, Seq is more efficient than EM.  Yet unlike EM, the number of iterations needed for Seq depends on $d$, since at least $d$ steps must be taken to fully recover $\bd z$. The overall time complexity is $O(d^2 D)$ and the space complexity is $O(d^2 + D)$ \cite{ji2008multitask}.  

However, Seq has several limitations.  
It is often the case that $d$ is a fraction or percentage of $D$.  If $d$ grows linearly with $D$, the asymptotic time cost is $O(D^3)$, similar to EM.  This may explain why Seq can still be up to hundreds of times slower than non-Bayesian methods for large $D$ \cite{bilgic2011multi}.  Also, the algorithm's sequential nature hinders its potential for speedup on parallel machines. Lastly, for large $D$, it remains costly to store parts of the quadratically-sized covariance matrix $\bd \Sigma$.                  

\section{Covariance-Free Expectation-Maximization} \label{sec:cofem}
We introduce \emph{covariance-free expectation-maximization} (CoFEM), a new SBL inference scheme for accelerating EM that removes the need to compute (let alone invert) the covariance matrix $\bd \Sigma$.  Our method is built on several advances from numerical linear algebra  \cite{bekas2007estimator, hestenes1952methods}.

Our key observation is that not all elements of $\bd \Sigma$ are needed for the M-Step in Eq. \eqref{mstep}. 
Indeed, we only need $\bd \mu$ (which depends on $\bd \Sigma$ via Eq. \eqref{estep-params}) and the the diagonal elements of $\bd \Sigma$ to update $\bhat \alpha$.
Thus, we propose a simplified E-Step that can estimate the two desired quantities $\{\mu_j, \Sigma_{j, j}\}$ for all $j$ from \emph{solutions to linear systems}, thereby avoiding the need for matrix inversion. We can re-express Eq. \eqref{estep-params} for $\bd \mu$ as 
\begin{align}
\bd \Sigma^{-1} \bd \mu = \beta \bd \Phi^\top \bd y, \label{mu-lin-sys}
\end{align}
where $\bd \Sigma^{-1} = \beta \bd \Phi^\top \bd \Phi + \mathrm{diag}\{\bidx \alpha t\}$.  Thus, $\bd \mu$ is the solution $\bd x$ to the linear system $\mathbf A \bd x = \bd b$ for $\mathbf A := \bd \Sigma^{-1}$ and $\bd b := \beta \bd \Phi^\top \bd y$.

\subsection{Diagonal Estimation} \label{simp-estep}
We leverage a result from \cite{bekas2007estimator} to obtain the diagonal of $\bd \Sigma$.  

\begin{prop}
[Diagonal Estimation Rule] Let $\mathbf M \in \R^{D \times D}$ be a square matrix.  Let $\bd p_1, \bd p_2, \ldots, \bd p_K \in \R^D$ be random probe vectors, where each $\bd p_k$ comprises independent and identically distributed elements such that $\E[p_{k, j}] = 0$ for all $j = 1, \ldots, D$.  For each $\bd p_k$, let $\bd x_k = \bold M \bd p_k$.  Consider the estimator $\bd s \in \R^D$, where, for each $j = 1, \ldots, D$,
\begin{align}
s_j = \frac{\sum_{k=1}^K p_{k, j} \cdot x_{k, j}}{\sum_{k=1}^K p_{k, j}^2}. \label{single-diag-estimator}
\end{align}
Then, each $s_j$ is an unbiased estimator of $\mathrm M_{j, j}$.  
\label{der}           
\end{prop}

\begin{proof}
Consider $s_j$, the $j$-th element of $\bd s$.  We have
\begin{align}
s_j &= \frac{\sum_{k=1}^K \left( p_{k, j} \cdot \left(\sum_{j'=1}^D \mathrm{M}_{j, j'} \cdot p_{k, j'} \right)\right)}{\sum_{k=1}^K p_{k, j}^2} \nonumber \\
&= \mathrm{M}_{j, j} + \sum_{j' \neq j} \mathrm{M}_{j, j'}  \cdot \frac{\sum_{k=1}^K p_{k, j} \cdot  p_{k, j'}}{\sum_{k=1}^K p_{k, j}^2}. \label{diag-expansion}
\end{align}
Thus, $\E[s_j]$ is equal to the following:
\begin{align}
&\mathrm{M}_{j, j} + \sum_{j' \neq j} \mathrm{M}_{j, j'}  \cdot \left(\sum_{k=1}^K \underbrace{\E[p_{k, j'}]}_{=0} \cdot  \E \left[\frac{p_{k, j}}{\sum_{k=1}^K p_{k, j}^2} \right]\right),
\end{align} 
where we have applied the fact that the $j$ and $j'$ components of $\bd p_k$ are independent to arrive at a product of expectations.  Since $\E[p_{k, j'}] = 0$ for all $k$ and $j'$, it follows that $\bd \E[s_j] = \mathrm{M}_{j, j}$.   
\end{proof}

We apply this rule to $\bd \Sigma$ to estimate its diagonal elements.  Following \cite{bekas2007estimator}, we use the \emph{Rademacher distribution} for drawing the probe vector $\bd p_k$, where each $p_{k,j}$ is either $-1$ or $+1$ with equal probability. This simplifies Eq. \eqref{single-diag-estimator} to 
\begin{align}
s_j = \frac{1}{K} \sum_{k=1}^K p_{k, j} \cdot x_{k, j} \label{diag-estimate}.
\end{align}        
The diagonal estimation rule turns an \emph{explicit} property of a matrix (i.e. the diagonal elements) into an \emph{implicit} quantity that can be estimated without physically forming the matrix.  To exploit this rule, we only need a method for applying $\bd \Sigma$ to each vector $\bd p_k$ to obtain $\bd x_k$; this can be done by solving a linear system $\bold A \bd x = \bd b$, where $\bold A := \bd \Sigma^{-1}$ and $\bd b := \bd p_k$.

In summary, the quantities $\{\mu_j, \Sigma_{j, j}\}$ needed for the simplified E-Step update can be obtained by solving $K + 1$ separate linear systems.  These systems can be solved in parallel by considering the matrix equation $\bold A \bold X = \bold B$ with inputs $\bold A \in \R^{D \times D}$ and  $\bold B \in \R^{D \times (K+1)}$ defined as follows:  
\begin{align}
\bold A &:= \beta \bd \Phi^\top \bd \Phi + \mathrm{diag}\{\bidx \alpha t\}, \nonumber \\
\bold B &:= 
\begin{bmatrix}
\bd p_1 \given \bd p_2 \given \ldots \given \bd p_K \given \beta \bd \Phi^\top \bd y \label{mat-eq}
\end{bmatrix}.
\end{align}  
If we enumerate the columns of the solution matrix $\bold X \in \R^{D \times (K + 1)}$ as $\bd x_1, \bd x_2, \ldots, \bd x_K, \bd \mu$, then our desired quantities for the simplified E-Step are $\bd \mu$ and $\bd s$, as calculated by Eq. \eqref{diag-estimate}.  We can then perform the M-Step update in Eq. \eqref{mstep} as
\begin{align}
\idx[j] \alpha {t+1} = \frac{1}{\mu_j^2 + s_j},
\end{align}
completely avoiding the need to compute or invert the covariance matrix $\bd \Sigma$.  Algorithm \ref{cofem} gives the full CoFEM algorithm.    

The diagonal estimator in Eq. \eqref{diag-estimate} is unbiased for any ${K \geq 1}$, yet its variance is proportional to $1 / K$, as we will show in Section \ref{probe-vec}.  We will provide theoretical justification that small $K$ suffices for CoFEM in practice, and that this $K$ can be constant with respect to the dimensionality $D$ for a large class of compressed sensing matrices $\bd \Phi$.      

%

\begin{algorithm}[!t!]
\caption{\textsc{CovarianceFreeEM}($\bd y$, $\bd \Phi$, $\beta$, $T$, $K$)} \label{cofem}
\begin{algorithmic}[1]
\State{Initialize $\alpha_j^{(1)} \gets 1$ for $j = 1, \ldots, D$.}
\For {$t = 1, 2, \ldots, T$}
    \State{\emph{// Simplified E-Step }}
    \State{Define $\bold A \gets \beta \bd \Phi^\top \bd \Phi + \mathrm{diag}\{\bidx \alpha t\}$.}
    \State{Draw $\bd p_1, \bd p_2, \ldots, \bd p_K \sim $ Rademacher distribution.}
    \State{Define $\bold B \gets [\bd p_1 \given \bd p_2 \given \ldots \given \bd p_K \given \beta \bd \Phi^\top \bd y]$.}
    \State{$[\bd x_1 \given \bd x_2 \given \ldots \given \bd x_K \given \bd \mu] \gets $ \textsc{LinearSolver}($\bold A, \bold B$).}
    \State{Compute $s_j \gets 1 / K\sum_{k=1}^K p_{k, j} \cdot x_{k, j}$ for $j = 1, \ldots, D$.}
    \State{\emph{// M-Step}}
    \If{$t < T$}
    \State {Update $\alpha_j^{(t + 1)} \gets 1 / (\mu_j^2 + s_j)$ for $j = 1, \ldots, D$.}
    \EndIf
\EndFor \\
\Return {$ \bidx \alpha T, \bd \mu, \bd s$}
\end{algorithmic}
\end{algorithm}

\subsection{Parallel Conjugate Gradient} \label{parallel-cg-alg}
Among potential options for the linear solver in Algorithm \ref{cofem}, we choose the conjugate gradient (CG) algorithm due to its efficiency and flexibility \cite{hestenes1952methods, shewchuk1994introduction}.  CG is an iterative approach with a series of matrix-vector multiplication steps to solve a linear system $\bold A \bd x = \bd b$.   At each step, CG does not need a physical manifestation of $\bold A$; it only requires a way to apply $\bold A$ to an arbitrary vector $\bd v$ to yield $\bold A \bd v$.  
For SBL, $\bold A := \beta \bd \Phi^\top \bd \Phi + \mathrm{diag}\{\bidx \alpha t\}$, which implies the complexity of CG is dominated by the time $\tau_D$ it takes to apply $\bd \Phi$ and its transpose to $\bd v$.
In many applications, $\bd \Phi$ is a structured and matrix-free operation.  Examples include convolution, discrete cosine transform, Fourier transform, and wavelet transform -- all of which require at most $\tau_D = O(D \log D)$-time~\cite{oppenheim2001discrete}.

CG can also easily generalize to solving \emph{multiple} linear systems $\bold A \bold X = \bold B$ by simply replacing the matrix-vector multiplications with matrix-matrix multiplications.  For faster computing, these operations can be parallelized on GPUs. In addition, CG is space-efficient and only needs $O(D)$-space to solve the linear system;  this is the minimum requirement for any solver given that the output $\bd x \in \R^D$.  

\subsection{Preconditioning} \label{sec:preconditioning}
The time complexity of CG depends on the number of CG steps $U$. It is guaranteed that $U\leq D$~\cite{hestenes1952methods}, yet $D$ can be very large for high-dimensional problems. In general, if $\bold A$ has a small condition number $\kappa(\bold A) := \lambda_\text{max}(\bold A) / \lambda_\text{min}(\bold A)$ where $\lambda_\text{max}(\bold A)$ and $\lambda_\text{min}(\bold A)$ are the largest and smallest eigenvalues of $\bold A$, then $U \ll D$ steps are needed to find an $\bhat x$ such that ${\norm{\bold A \bhat x - \bd b}_2 / \norm{\bd b}_2 \leq \epsilon_\text{max}}$ for small $\epsilon_\text{max}$. However, optimizing the SBL objective function pushes many diagonal elements of $\bold A$ to $\infty$, resulting in large $\kappa(\bold A)$ and necessitating large $U$.

To resolve this issue, we incorporate a \emph{preconditioner} matrix $\bold M \in \R^{D \times D}$ in CG.  Instead of directly solving $\bold A \bd x = \bd b$, we solve an equivalent system $\bold A' \bd x' = \bd b'$, where $\bold A' := \bold M^{-1/2} \bold A \bold M^{-1/2}$, $\bd b' := \bold M^{-1/2} \bd b$ and $\bd x' := \bold M^{1/2} \bd x$. We consider two criteria for choosing $\bold M$ in CoFEM. First, we want $\kappa(\bold A') \ll \kappa(\bold A)$ to reduce the number of steps $U$ that are necessary to achieve small $\epsilon_\text{max}$. Next, we want to maintain the scalability of CoFEM (e.g. absence of matrix inversion, $O(D)$-space complexity) by avoiding a dense matrix for $\bold M$. 

We thus propose to use a \emph{diagonal preconditioner} $\bold M = \mathrm{diag}\{\beta \bd \theta + \bidx \alpha t\}$, where $\bd \theta \in \R^D$ is a set of customizable positive values.  This $\bold M$ is easy to invert, only requires $O(D)$-space, and can be quickly applied to vectors.  By varying $\bd \theta$, we can arrive at different choices for $\bold M$.  For example, setting $\theta_j = \sum_{i=1}^N \mathrm \Phi_{i, j}^2$ leads to the popular \emph{Jacobi preconditioner} satisfying $\mathrm M_{j, j} = \mathrm A_{j, j}$ \cite{shewchuk1994introduction}. In Section \ref{sec:theory-u}, we provide novel theoretical analysis for our diagonal preconditioner within the context of SBL.  We illustrate that for a large class of compressed sensing dictionaries $\bd \Phi$, setting $\theta_j = 1$ for all $j$ is a favorable choice that leads to small $\kappa(\bold A')$ and enables $U$ to be constant with respect to the dimensionality $D$.  

Algorithm \ref{cg} summarizes the parallel CG algorithm for inputs $\bold A \in \R^{D \times D}$, $\bold M \in \R^{D \times D}$, and $\bold B \in \R^{D \times Q}$, where $Q$ is the number of parallel systems.  For CoFEM, we have $Q = K + 1$.  The computation is dominated by line 5, in which $\bold A$ is applied to vectors stored as columns of a matrix.


\begin{algorithm}[!t!]
\caption{\textsc{ParallelConjGradient}($\bold A, \bold B, \bold M, U, \epsilon_\text{max}$)} \label{cg}
\begin{algorithmic}[1]
\State{Initialize $\bold X$ as a $D \times Q$ matrix of all zeros.} 
\State{Initialize $\bold R \gets \bold B$ and $\bold P \gets \bold B$ and $\bold W \gets \bold M^{-1} \bold B$.}
\State{Compute $\rho_q \gets \sum_{j=1}^D \mathrm R_{j, q} \cdot \mathrm W_{j, q}$ for $q = 1, \ldots, Q$.}
\For {$u = 1, 2, \ldots, U$}
    \State{Compute $\bd \Psi \gets \bold A \bold P$.  \emph{// Apply matrix.}}
    \State{Compute $\pi_q \gets \sum_{j=1}^D \mathrm P_{d, q} \cdot \Psi_{d, q}$ for $q = 1, \ldots, Q$.}
    \State{Compute $\gamma_q \gets \rho_q / \pi_q$ for $q = 1, \ldots, Q$.}
    \State{Update $\bold X \gets \bold X + \bold P \bold \Gamma$, where $\bold \Gamma = \mathrm{diag}\{\bd \gamma\}$.}
    \State{Update $\bold R \gets \bold R -  \bold \Psi \bold \Gamma$, where $\bold \Gamma = \mathrm{diag}\{\bd \gamma\}$.}
    \State{Let $\delta \gets ||\bold R||_F / ||\bold B||_F$, where $F$ is Frobenius norm.}
    \If{$\delta \leq \epsilon_\text{max}$}
    	\State{\Return{$\bold X$}}
    \EndIf
    \State{Compute $\bold W \gets \bold M^{-1} \bold R$. \emph{// Apply preconditioner.}}
    \State{Let $\rho^\text{old}_q \gets \rho_q$ for $q = 1, \ldots, Q$.}
    \State{Compute $\rho_q \gets \sum_{j=1}^D \mathrm R_{j, q} \cdot \mathrm W_{j, q}$ for $q = 1, \ldots, Q$.}
    \State{Compute $\eta_q \gets \rho_q / \rho^\text{old}_q$ for $q = 1, \ldots, Q$.}
    \State{Update $\bold P \gets \bold R + \bold P \bold H$, where $\bold H = \mathrm{diag}\{\bd \eta\}$.}
\EndFor \\
\Return {$ \bold X$}
\end{algorithmic}
\end{algorithm}

\subsection{Complexity Comparisons}
Each of the $T$ iterations of CoFEM requires at most $U$ steps of CG -- in which we apply $\bd \Phi$ (and $\bd \Phi^\top$) in $\tau_D$-time to $K$ vectors -- giving us an overall time complexity of $O(T \tau_D U K)$.  CoFEM's space complexity is dominated by CG, which requires $O(D)$-space for each of the $(K+1)$ systems. Table \ref{tab:complex} shows the complexities of CoFEM and other SBL inference schemes. While many other methods improve upon EM, they introduce dependencies on $N$ or $d$, which typically grow with $D$.  For example, if the size of the signal $\bd z$ is doubled, we may also expect the number of measurements $N$ to be doubled (to achieve same reconstruction error), as well as the number $d$ of non-zero values  in $\bd z$. Thus, increasing $D$ compounds the increase in complexities of these algorithms.  In contrast, CoFEM's dependencies on $U$ and $K$ can be held constant as $D$ increases, which we demonstrate in Section~\ref{sec:simulated}.

\begin{table}
\centering
\caption{Computational complexities of SBL inference schemes.}\label{tab:complex}
\begin{tabular}{lcc}
\toprule
\textbf{Method} & \textbf{Time} & \textbf{Space}  \\ 
\hline
EM \cite{tipping2001sparse} & $O(T D^3)$ & $O(D^2)$ \\
IRLS \cite{wipf2010iterative} & $O(T (DN^2 + N^3))$ & $O(D^2)$\\
AMP \cite{fang2016two} & $O(T  DNT_\text{amp})$ & $O(DN)$ \\
VI \cite{duan2017fast} & $O(T \tau_D)$ & $O(D)$ \\
Seq \cite{tipping2003fast} & $O(Dd^2)$ & $O(D + d^2)$ \\
CoFEM (ours) & $O(T \tau_D U K)$ & $O(DK)$ \\
\bottomrule
\end{tabular}
\end{table}

\subsection{CoFEM for SBL Extensions} \label{sec:sbl-extensions-cofem}
To further highlight the flexibility of CoFEM, we show how it can handle two common extensions of the SBL model.


\subsubsection{Multi-Task Learning} \label{multi-task-cofem}
{In multi-task learning, there are $L$ different sparse vector recovery problems that one wishes to solve at once.  These problems may have different observation-dictionary pairs $(\bd y_1, \bd \Phi_1), (\bd y_2, \bd \Phi_2), \ldots, (\bd y_L, \bd \Phi_L)$, yet the tasks are related in the sense that their corresponding vectors $\bd z_1, \bd z_2, \ldots, \bd z_L$ have {similar non-zero supports}.  Some examples include multiple measurements vector (MMV) problems \cite{wipf2007empirical}, multi-task compressed sensing \cite{ji2008multitask}, and sparse Bayesian learning with complex numbers \cite{wu2014complex}.

A simple way to enforce joint sparsity among all tasks in SBL is to have them share a common $\bd \alpha$ vector:
\begin{align}
\bd z_\ell &\sim \mathcal{N}(\bd 0, \mathrm{diag}\{\bd \alpha\}^{-1}), & \ell = 1, 2, \ldots, L, \nonumber \\
\bd y_\ell &\sim \mathcal{N}(\bd \Phi_\ell \bd z_\ell, 1 / \beta \mathbf I), & \ell = 1, 2, \ldots, L. \label{multi-model}
\end{align}
Learning takes place through the task-separable objective:
\begin{align}
\max_{\bd \alpha} \log p(\bd y_1, \bd y_2, \ldots, \bd y_L \given \bd \alpha) = \sum_{\ell = 1}^L \log p(\bd y_\ell\given \bd \alpha). \hspace{-0.3em} \label{multi-mle2}
\end{align}
}

To optimize Eq. \eqref{multi-mle2}, EM runs a E-Step for each task $\ell$,
\begin{align}
\bd \mu_\ell = \beta \bd \Sigma_\ell \bd \Phi_\ell^\top \bd y_\ell, & & \bd \Sigma_\ell = (\beta \bd \Phi_\ell^\top \bd \Phi_\ell + \text{diag}\{\bidx \alpha t\})^{-1},
\end{align}   
followed by a M-Step to combine these moments,
\begin{align}
\idx[j] \alpha {t+1} = \frac{L}{\sum_{\ell = 1}^L \mu_{\ell, j}^2 + \Sigma_{\ell, j, j}}.
 \end{align}

To accelerate EM, CoFEM can simply replace the E-Step for each task $\ell$ with a covariance-free version, as described in Section \ref{simp-estep}.  We can further parallelize these $L$ E-Steps by solving their $L \cdot (K+1)$ systems all at once through Alg. \ref{cg}.   

\subsubsection{Non-Negativity Constraints} \label{sec:sbl-nonneg-cofem} 
{
In some applications, we may want to enforce non-negativity on $\bd z$. In these cases,
we can use an independent \emph{rectified Gaussian} prior $\mathcal{N}^R(0, 1 / \alpha_j)$ for each component $z_j$ of $\bd z$ \cite{nalci2018rectified}, which places zero probability mass on the negative values. Specifically, we have 
\begin{align}
p(z_j \given \alpha_j) = \begin{cases}
\sqrt{2 \alpha_j / \pi} \exp(- \alpha_j z_j^2 / 2), & z_j > 0,\\ 
1/2, &z_j = 0,\\
0, & z_j < 0. 
\end{cases}
\end{align}


}

Due to conjugacy between the rectified Gaussian and Gaussian distributions, the posterior $p(\bd z \given \bd y, \bhat \alpha)$ is also a rectified Gaussian.  However, this posterior's density function is not analytically tractable, so we can follow \cite{nalci2018rectified} and approximate it with a diagonal rectified Gaussian whose second moment is 
 \begin{align}
 \E_{p(\bd z \given \bd y, \bhat \alpha)}[z_j^2] = \mu_j^2 + \Sigma_{j, j} + \mu_j \cdot \sqrt{\frac{\Sigma_{j, j}}{\pi}} \cdot \frac{\exp(- \xi_j^2 )}{\mathrm{erfc} (- \xi_j )}, \label{nonneg-sbl}
 \end{align}
 where $\xi_j = \mu_j / \sqrt{2\Sigma_{j, j}}$ and $\mathrm{erfc}(x) = 2 / \sqrt{\pi}\int_{x}^\infty\exp(-t^2) dt$ is the complimentary error function.  
To compute Eq. \eqref{nonneg-sbl}, CoFEM can replace $\Sigma_{j, j}$ with $s_j$ from Eq. \eqref{diag-estimate}.  The M-Step then updates $\alpha_j$ as the reciprocal of Eq. \eqref{nonneg-sbl}.

\section{Theoretical Analysis of CoFEM} \label{sec:theory}
The two main hyperparameters of CoFEM are the number of probe vectors $K$ and the number of conjugate gradient steps $U$.  These values determine the per-iteration time complexity of CoFEM as $O(UK\tau_D)$ and its space complexity as $O(DK)$.  Thus, it is important to control $U$ and $K$.  In this section, we present new theoretical results illustrating that, for a large class of dictionaries $\bd \Phi$, these hyperparameters can be kept small even as the dimensionality of the problem $D$ grows very large.  

The ultimate goal of SBL is to solve a Bayesian variant of the sparse coding problem in which we (a) identify which $z_j = 0$ and (b) provide uncertainty quantification for the non-zero $z_j$.  To achieve true sparsity for a particular $z_j$, it is necessary for its prior parameter $\alpha_j \to \infty$, which is only attained as the number of iterations $t \to \infty$.  Thus, our study of $U$ and $K$ is based on the following SBL Convergence property:

\begin{definition}[SBL Convergence] \label{sbl-conv-def}
In Alg. \ref{cofem}, consider the sequence of iterates $\bidx \alpha t$ for $t = 1, 2, \ldots,$ and let $\bhat \alpha := \lim_{t \to \infty} \bidx \alpha t$.  Then, $(\mathcal{S}, \mathcal{U}, \bhat \alpha)$-convergence is satisfied for SBL if the indices $\mathbb{N}_D := \{1, 2, \ldots, D\}$ can be partitioned into an ``active" set $\mathcal{S} \subseteq \mathbb{N}_D$ and an ``inactive" set $\mathcal{U} := \mathbb{N}_D \setminus \mathcal{S}$, where $\hat{\alpha}_j > 0$ is finite if $j \in \mathcal{S}$ and $\hat{\alpha}_{j} = \infty$ if $j \in \mathcal{U}$.  
\end{definition}      

Of course, it is impossible to run $t \to \infty$ iterations in practice and to the best of the authors' knowledge, there does not exist a formal proof characterizing EM's behavior after finitely many iterations in the context of SBL.  However, many works within the SBL literature \cite{tipping2001sparse, wipf2004sparse, ji2008bayesian, ji2008multitask} have illustrated (and even relied on) a well-accepted phenomenon in which the inactive $\alpha_{j}$ grow very large and can be treated as ``reaching infinity" after $T$ finite iterations.  This justifies the practical applicability of Definition \ref{sbl-conv-def} and our ensuing theoretical analysis. 

\noindent\textbf{Notation} For any positive integer $P$, let $\mathbb{N}_P := \{1, 2, \ldots, P\}$.  For any vector $\bd v \in \R^P$ and set $\mathcal{B} \subseteq \mathbb{N}_P$, define the sub-vector $\bd v_\mathcal{B} := [v_b \given b \in \mathcal{B}] \in \R^{|\mathcal{B}|}$.  Similarly, for any matrix $\bold M = [\bidx v 1, \ldots, \bidx v P] \in \R^{Q \times P}$, define the sub-matrix $\bold M_\mathcal{B} := [\bidx v b \given b \in \mathcal{B}] \in \R^{Q \times |\mathcal{B}|}$.  For any two sets $\mathcal{A} \subseteq \N_Q$ and $\mathcal{B} \subseteq \N_P$, define the sub-matrix block $\bold M_{\mathcal{A}, \mathcal{B}} := [\bidx[{\mathcal{A}}] v b \given b \in \mathcal{B}] \in \R^{|\mathcal{A}| \times |\mathcal{B}|}$.  Let $\norm{\bd v}_2$ denote the Euclidean norm of vector $\bd v$ and $\norm{\bd v}_{\bold M} := \sqrt{\bd v^\top \bold M \bd v}$ denote the $\bold M$-weighted norm of $\bd v$.  Let $\norm{\bold M}_2$ denote the spectral norm (i.e. largest singular value) of matrix $\bold M$.  Let $\sigma_\text{min}(\bold M), \lambda_\text{max}(\bold M)$ and $\lambda_\text{min}(\bold M)$ denote the smallest singular value, largest eigenvalue, and smallest eigenvalue of $\bold M$, respectively.  For a matrix $\bold M$, let $\kappa(\bold M) := \norm{\bold M}_2 / \sigma_\text{min}(\bold M)$  be the condition number of $\bold M$.  

\subsection{A Theory for the Number of Probe Vectors $K$} \label{probe-vec}
First, we analyze the dependency of the diagonal estimator $\bd s$ on the number of probe vectors $K$.
We aim to characterize the standard deviation (i.e. the square root of the variance) of each $s_j$, which leads to this lemma:
\begin{lemma}
 \label{diag-var-lemma}
Let $\bold M \in \R^{D \times D}$.  Consider the estimator $\bd s$ defined in Prop. \ref{der}, where $\bd p_1, \ldots, \bd p_K$ are independent Rademacher variables.  Then, the standard deviation $\nu_j$ of $s_j$ is 
\begin{align}
\nu_j := \sqrt{\E[(s_j - \E[s_j])^2]} = \sqrt{\frac{1}{K} \sum_{j' \neq j} \mathrm M_{j, j'}^2}. \label{diag-var}
\end{align}
\end{lemma}
\begin{proof}
Within Eq. \eqref{diag-var}, we substitute the expression for $s_j$ from Eq. \eqref{diag-expansion} and the fact that $\E[s_j] = \mathrm M_{j, j}$ to yield
\begin{align}
\nu_j &= \sqrt{\E\left[\left(\sum_{j' \neq j} \mathrm{M}_{j, j'}  \cdot \frac{\sum_{k=1}^K p_{k, j} \cdot  p_{k, j'}}{\sum_{k=1}^K p_{k, j}^2}\right)^2\right]} \nonumber \\
&= \sqrt{\E\left[\sum_{j' \neq j} \sum_{j'' \neq j} \mathrm{M}_{j, j'} \cdot \mathrm{M}_{j, j''} \cdot \frac{g_{j, j'}}{g_{j, j}} \cdot \frac{g_{j, j''}}{g_{j, j}}\right]},\label{diag-var-final}
\end{align}
where we define $g_{j, \ell} := \sum_{k=1}^K p_{k, j} \cdot p_{k, \ell}$ for all $(j, \ell) \in \N_D \times \N_D$.  In the denominator, we have $g_{j, j} = K$ for all $j$ since $p^2 = 1$ for a Rademacher variable.  In the numerator, if $j' = j''$, we have $\E[g_{j, j'} \cdot g_{j, j''}] = \E[g_{j, j'}^2] = K$.  Otherwise, if $j \neq j''$, $\E[g_{j, j'} \cdot g_{j, j''}] = 0$. Thus, Eq. \eqref{diag-var-final} simplifies to Eq. \eqref{diag-var}.  
\end{proof}

Lemma \ref{diag-var-lemma} tells us that the standard deviation of our estimator decreases with $K$ and increases with the norm of the off-diagonal entries.  Analyzing Lemma \ref{diag-var-lemma} within the context of SBL leads to our first main theoretical result, as stated below.
\begin{theorem}
\label{var-bound-thm}
Let $\bidx \Sigma t := (\beta \bd \Phi^\top \bd \Phi + \mathrm{diag}\{\bidx \alpha t\})^{-1}$ be the SBL covariance matrix at the $t$-th iteration of Alg. \ref{cofem}.  Let $\bidx s t$ be the Rademacher diagonal estimator for $\bidx \Sigma t$ defined in Eq. \eqref{diag-estimate} with $K$ probe vectors.  Let $\idx[j] \nu t$ denote the standard deviation of $\idx[j] s t$.  Assume that $(\mathcal{S}, \mathcal{U}, \bhat \alpha)$-convergence is satisfied.  Then, for any inactive index $j \in \mathcal{U}$, we have 
\begin{align}
\lim_{t \to \infty} \idx[j] \nu t  = 0, \label{inactive-var}
\end{align}   
and for any active index $j \in \mathcal{S}$, we have
\begin{align}
\lim_{t \to \infty} \idx[j] \nu t  \leq \frac{1}{\sqrt{K}} \cdot \frac{\inf_{\bd \Theta \in \mathcal{O}} \norm{\bd \Theta^{-1} \bd \Phi^\top_{\mathcal{S}} \bd \Phi_\mathcal{S} - \bold I}_2}{\beta \cdot \sigma^2_\mathrm{min}(\bd \Phi_\mathcal{S})}, \label{active-var}
\end{align}   
where $\mathcal{O}$ is the set of $|\mathcal{S}| \times |\mathcal{S}|$ diagonal matrix with positive diagonal elements and $\bold I$ is the identity matrix.
\end{theorem}

Theorem \ref{var-bound-thm} offers several insights in the limit of EM iterations: (1) the estimator becomes deterministic with zero standard deviation for the inactive indices, (2) $K$ only affects the estimator's standard deviation for the active indices, and (3) if $\bd \Phi^\top_\mathcal{S} \bd \Phi_\mathcal{S}$ is close to any diagonal matrix $\bd \Theta$ (i.e. the columns of $\bd \Phi_\mathcal{S}$ are close to orthogonal), then the standard deviation for the active indices converge to a small quantity.  The proof of Theorem \ref{var-bound-thm} is given in Appendix \ref{var-bound-proof}.   

\subsection{A Theory for the Number of Conjugate Gradient Steps $U$} \label{sec:theory-u}
Next, we analyze the number of CG steps $U$ needed for convergence. We build on the following well-known result~\cite{shewchuk1994introduction}:
\begin{lemma}[CG Convergence] 
Consider the CG algorithm for solving $\bold A \bd x = \bd b$, where $\bold A \in \R^{D \times D}$ is a positive definite matrix and $\bd b \in \R^D$.  Let $\bd x_0 \in \R^D$ be the initial solution.  Let $\bd x_u$ denote the algorithm's solution and $\bd r_u := \bd b - \bd A \bd x_u$ denote the algorithm's residual at the $u$-th step of CG.  Then, 
\begin{align}
\norm{\bd r_u}_{\bold A^{-1}} \leq 2 \left(\frac{\sqrt{\kappa(\bold A)} - 1}{\sqrt{\kappa(\bold A)} + 1}\right)^{u} \norm{\bd r_0}_{\bold A^{-1}}, \label{cg-result}
\end{align}
where $\norm{\bd r}_{\bold A^{-1}} := \sqrt{\bd r^\top \bold A^{-1} \bd r}$ for any vector $\bd r \in \R^D$.
\end{lemma}

\begin{corollary}
\label{coro-epsilon}
Let CG with matrix $\bold A$ and any $\bd b \in \R^D$ start with the initialization $\bd x_0 = \bd 0$.  Let $\epsilon := \norm{\bd r_U}_{\bold A^{-1}} /\norm{\bd b}_{\bold A^{-1}}$ denote the relative residual error of CG after $U$ steps\footnote{Theoretical results for CG often define $\epsilon$ in terms of the norm weighted by $\bold A^{-1}$ (or $\bold A$), even though $\epsilon_2 := \norm{\bd r_U}_2/\norm{\bd b}_2$ is used to determine convergence in software.  We follow this convention while noting that one can exploit the relation $\epsilon_2 \leq \epsilon \sqrt{\kappa(\bold A)}$ to extend our results for $\epsilon_2$.}.  Then, 
\begin{align}
\epsilon \leq 2 \exp(-U / \sqrt{\kappa(\bold A)}). \label{epsilon-bound}
\end{align}  
\end{corollary}
\begin{proof}
Since $\bd r_0 = \bd b$, we can apply Eq. \eqref{cg-result} to obtain
\begin{align}
\epsilon \leq 2 \left(\frac{\sqrt{\kappa(\bold A)} - 1}{\sqrt{\kappa(\bold A)} + 1}\right)^{U} \leq 2 \left(1 - \frac{1}{\sqrt{\kappa(\bold A)}}\right)^{U}. \label{ineq}
\end{align}
The inequality $1 - 1/x \leq \exp(-1/x)$ holds for any $x \in \R$.  Taking $x = \sqrt{\kappa(\bold A)}$ in Eq. \eqref{ineq} gives the result.
\end{proof}

Corollary \ref{coro-epsilon} indicates that the relative residual error of CG decreases exponentially with $U$, yet the precise exponent depends on $\kappa(\bold A)$.  In the $t$-th iteration of CoFEM, we wish to solve linear systems with $\boldidx A t := \beta \bd \Phi^\top \bd \Phi + \mathrm{diag}\{\bidx \alpha t\}$. However, since $\bd \alpha_j^{(t)}\rightarrow \infty$ for $j\in \mathcal{U}$, this leads to $\kappa(\boldidx A t) \to \infty$ and a vacuous bound on $\epsilon$ for any finite number of steps $U$.
Therefore, CoFEM introduces a preconditioner $\bold M$ to construct a new CG matrix $\boldidx {A'} {t}$ with a reduced condition number.  In our second main theoretical result, we analyze $\kappa(\boldidx {A'} t)$, the bound that it induces on error, and the implications for $U$.
\begin{theorem} \label{kappa-bound}
Let $\boldidx A t := \beta \bd \Phi^\top \bd \Phi + \mathrm{diag}\{\bidx \alpha t\}$ denote the SBL inverse-covariance matrix at the $t$-th iteration of Alg. \ref{cofem}.  Let $\boldidx M t := \mathrm{diag}\{\beta \bd \theta + \bidx \alpha t\}$ denote the preconditioner, where $\bd \theta \in \R^D$ is a vector of positive values.  Define the preconditioned matrix $\boldidx {A'} t := (\boldidx M t)^{-1/2} \boldidx A t (\boldidx M t)^{-1/2}$.  Let $\bidx b t \in \R^D$ be any vector and $\bidx {b'} t := (\boldidx M t)^{-1/2} \bidx b t$.  Let $\idx \epsilon t$ be the relative residual error after running $U$ conjugate gradient steps to solve the system $\boldidx {A'} t \bidx {x'} t = \bidx {b'} t$ with $\bidx[0] {x'} t = \bd 0$.    
Given $(\mathcal{S}, \mathcal{U}, \bhat \alpha)$-convergence, it follows that  
\begin{align}
\lim_{t \to \infty} \idx \epsilon t \leq 2\exp \left(-U \sqrt{\frac{1 - \norm{\bd \Theta^{-1} \bd \Phi_\mathcal{S}^\top \bd \Phi_\mathcal{S} - \bold I}_2}{1 + \norm{\bd \Theta^{-1} \bd \Phi_\mathcal{S}^\top \bd \Phi_\mathcal{S} - \bold I}_2}}\right), \label{epsilon-eq}
\end{align}  
where $\bold \Theta = \mathrm{diag}\{\bd \theta_\mathcal{S}\}$ and $\bold I$ is the identity matrix.
\end{theorem}

 Eq. \eqref{epsilon-eq} indicates that faster convergence is achieved for $\bd \theta_\mathcal{S} \in \R^{|\mathcal{S}|}$ that minimizes $\norm{\mathrm{diag}\{\bd \theta_\mathcal{S}\}^{-1} \bd \Phi_\mathcal{S}^\top \bd \Phi_\mathcal{S} - \bold I}_2$.  In practice, the support $\mathcal{S}$ is not known in advance, so we may instead choose $\bd \theta \in \R^{D}$ to minimize $\norm{\mathrm{diag}\{\bd \theta\}^{-1} \bd \Phi^\top \bd \Phi - \bold I}_2$.  For example, if $\bd \Phi^\top \bd \Phi$ is a diagonal matrix (i.e. all columns of $\bd \Phi$ are mutually orthogonal), then the optimal choice is $\theta_{j} = \sum_{i=1}^N \mathrm \Phi_{i, j}^2$, which corresponds to the Jacobi preconditioner.  The proof of Theorem \ref{kappa-bound} is given in Appendix \ref{kappa-bound-proof}.

\subsection{A Theory for $U$ and $K$ for Compressed Sensing Matrices}

We now show that our bounds in Theorems \ref{var-bound-thm} and \ref{kappa-bound} can be simplified for a large class of matrices $\bd \Phi$ satisfying the restricted isometry property (RIP) and commonly employed in compressed sensing applications.  

\begin{definition}[Restricted Isometry Property] 
\label{rip}
Let $\bd \Phi \in \R^{N \times D}$, $d \leq D$, and $\delta > 0$.  Then, $\bold \Phi$ satisfies $(d, \delta)$-RIP if for every set $\mathcal{C} \subseteq \{1, \ldots, D\}$ of size $|\mathcal{S}| = d$ and every vector $\bd v \in \R^d$,       
\begin{align}
(1 - \delta) \norm{\bd v}_2^2 \leq \norm{\bd \Phi_\mathcal{C} \bd v}_2^2 \leq (1 + \delta) \norm{\bd v}_2^2. \label{rip-eq}
\end{align}
\end{definition} 

\begin{corollary} 
\label{rip-coro}
Let $\bd \Phi \in \R^{N \times D}$ satisfy $(d, \delta)$-RIP.  For any $\mathcal{C} \subseteq \N_D$ with $|\mathcal{C}| = d$, let $\bd \Delta_{\mathcal{C},\mathcal{C}}  := {\bd \Phi_\mathcal{C}^\top \bd \Phi_\mathcal{C}} - \bold I$.  Then, $\norm{ \bold \Delta_{\mathcal{C},\mathcal{C}}}_2 \leq \delta$.
\end{corollary}

\begin{proof}
From Eq. \eqref{rip-eq}, we have for any vector $\bd v \in \R^d$,
\begin{align}
1 - \delta &\leq \frac{\bd v^\top \bd \Phi_\mathcal{C}^\top \bd \Phi_\mathcal{C} \bd v }{\bd v^\top \bd v} \leq 1 + \delta \implies \left| \frac{\bd v^\top \bold \Delta_{\mathcal{C}, \mathcal{C}} \bd v }{\bd v^\top \bd v} \right| \leq \delta. \label{spect-bound}
\end{align}
Eq. \eqref{spect-bound} bounds the Rayleigh quotient (and eigenvalues) of $\bold \Delta_{\mathcal{C}, \mathcal{C}}$ between $[-\delta, \delta]$.  The spectral norm of a symmetric matrix is equal to its largest absolute eigenvalue.
\end{proof}

\begin{corollary} \label{coro-k} In Theorem \ref{var-bound-thm}, let $\bd \Phi  \in \R^{N \times D}$ satisfy $(d, \delta)$-RIP, where $d = |\mathcal{S}|$.  Then, for $j \in \mathcal{S}$, the standard deviation of the diagonal estimator $\idx[j] \nu t$  satisfies
\begin{align}
\lim_{t \to \infty} \idx[j] \nu t  \leq \frac{1}{\sqrt{K}} \cdot \frac{\delta}{\beta (1 - \delta)}.
\end{align}
\end{corollary}
\begin{proof}
In Eq. \eqref{active-var}, take $\bd \Theta = \bold I$, which reduces the numerator to $\norm{\bold \Phi^\top_\mathcal{S} \bold \Phi_\mathcal{S} - \bold I}_2$.  By Corollary \ref{rip-coro}, this quantity is at most $\delta$.  Finally, by Definition \ref{rip}, $\sigma_\text{min}^2(\bd \Phi_\mathcal{S}) \geq (1 - \delta)$.   
\end{proof}

\begin{corollary} \label{coro-u} In Theorem \ref{kappa-bound}, let $\bd \Phi \in \R^{N \times D}$ satisfy $(d, \delta)$-RIP, where $d = |\mathcal{S}|$.  Let $\theta_j = 1$ for all $j \in \N_D$.  Then, the relative residual error $\idx \epsilon t$ of conjugate gradient satisfies
\begin{align}
\lim_{t \to \infty} \idx \epsilon t \leq 2\exp \left(-U\sqrt{\frac{1-\delta}{1+\delta}}\right). 
\end{align}  
\end{corollary}
\begin{proof}
This follows from applying Corollary \ref{rip-coro} to Eq. \eqref{epsilon-eq}.
\end{proof}

One may have expected that as $N$ and $D$ increase, CoFEM's hyperparameters $U$ and $K$ must also increase proportionally to ensure that $\nu_j$ and $\epsilon$ remain small.
However, Corollaries \ref{coro-k} and \ref{coro-u} illustrate that for RIP matrices\footnote{RIP is one mathematical notion for the idea of ``close to orthonormality" for a set of dictionary columns.  There exist other notions (e.g. incoherence, null space property) \cite{foucart2013}.  We conjecture there are bounds similar to Corollaries \ref{coro-k} and \ref{coro-u} that hold for matrices that satisfy these other properties.}, the bounds on $U$ and $K$ only depend on $\delta$ (not $N$ or $D$).  In compressed sensing, $\delta$ can be small even as $N$ and $D$ grow very large \cite{foucart2013}.  In Section \ref{sec:simulated}, we use this insight to demonstrate that even for very large $D$, CoFEM can accurately perform sparse coding with small constant values for $U$ and $K$.


\section{Simulated Experiments} \label{sec:simulated}
In this section, we run a series of experiments to compare the accuracy and scalability of CoFEM to that of other SBL inference methods across a broad range of different settings.  

\begin{figure*}
\begin{center}
\includegraphics[scale=0.47]{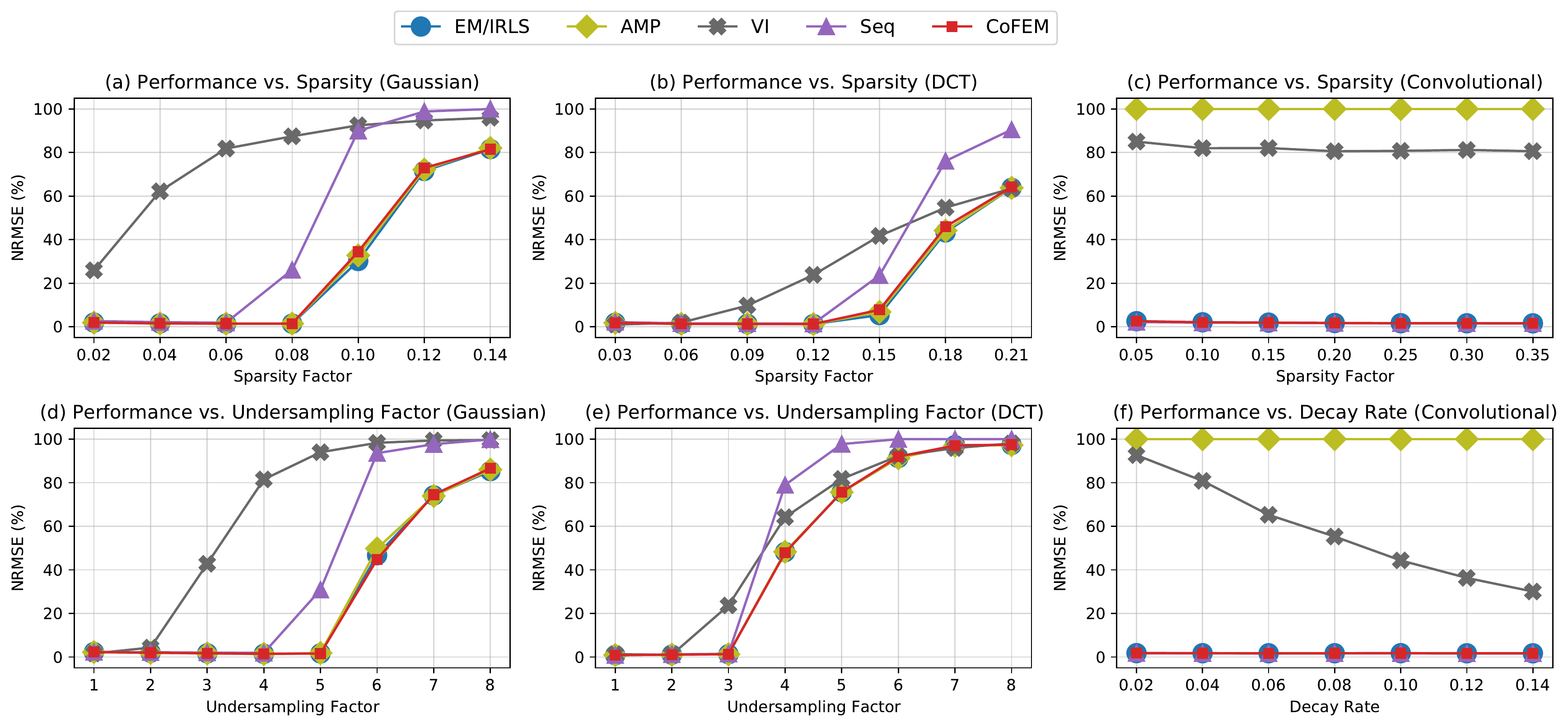} 
\caption{Comparing the accuracy of different SBL inference schemes.  Each point represents the mean of 25 trials.} \label{acc}
\end{center}
\end{figure*}

\subsection{Experimental Setup} \label{setup}
\subsubsection{General Structure} 
We design all simulations with a structure inspired by Section V-A of \cite{ji2008bayesian}.  
We form a ground-truth latent vector $\bd z^* \in \R^D$ of spikes by drawing $d$ of its components from a distribution $\mathcal{P}$ and setting the other $D - d$ components to zero.  The location of the spikes are chosen uniformly at random.  Given a dictionary $\bd \Phi \in \R^{N \times D}$, the observed data $\bd y \in \R^N$ is generated as $\bd y = \bd \Phi \bd z^* + \bd \varepsilon$, where $\bd \varepsilon \sim \mathcal{N}(0, \sigma^2\mathbf{I})\in \R^N$ with $\sigma = 0.01$.  The goal is to apply SBL for reconstructing $\bd z^*$.  Success is measured through minimization of normalized root mean squared error (NRMSE), $||\bhat z - \bd z^*||_2/||\bd z^*||_2 \times 100$,
where $\bhat z = \bd \mu$ is the mean of the distribution $p(\bd z \given \bd y, \bhat \alpha)$ upon convergence.  

\subsubsection{Dictionary} \label{sec:dict-types}
We consider three types of dictionaries $\bold \Phi$:
\begin{itemize}
\item \textbf{Dense Gaussian}: We draw each element $\Phi_{i, j}$ independently from $\mathcal{N}(0, 1/ N)$ to form a dense matrix $\bd \Phi$.  The spike distribution of $\bd z^*$ is $\mathcal{P} := \mathrm{Uniform}(-2, 2)$.
\item \textbf{DCT}: We let $\bd \Phi = \bold M \bold \Omega^{-1}$, where $\bold \Omega \in \R^{D \times D}$ is the matrix corresponding to the 1D discrete cosine transform of size $D$ and $\bold M \in \R^{N \times D}$ is an undersampling operator that selects $N$ out of $D$ components (where $N \leq D$).  The spike distribution is $\mathcal{P} := \mathcal{N}(0, 5)$.
\item \textbf{Convolutional}: We set $N = D$ and $\bd \Phi \in \R^{D \times D}$ to a convolution in which its columns are delayed (and truncated) repetitions of an exponentially decaying filter $\bd \phi \in \R^D$.  That is, $\phi_j := (1 - \rho)^{j-1}$ for $0<\rho<1$.  Thus, $\bd \Phi$ is a lower triangular matrix in which the $j$-th column is a concatenation of $j-1$ zeros and $\{\phi_1, \ldots, \phi_{N-(j-1)}\}$.   The spike distribution is $\mathcal{P} := \mathrm{Exponential}(1.5)$.
\end{itemize}

\subsection{Accuracy Analysis} \label{sec:cofem-em}
For our first analysis, we fix $D = 1024$ and evaluate the accuracy of various SBL inference schemes across different settings.  We compare among EM/IRLS \cite{tipping2001sparse}, AMP \cite{fang2016two}, VI \cite{duan2017fast}, Seq \cite{tipping2003fast}, and CoFEM. Note that EM and IRLS always yield the same result; IRLS is just the Woodbury identity (Eq. \eqref{woodbury}) applied to EM.  All EM-based methods (EM/IRLS, AMP, VI, CoFEM) are executed for $T = 50$ iterations.  AMP employs $T_\text{amp} = 10$ inner loops.  For CoFEM, Corollary \ref{coro-k} tells us that we can keep the number of probes $K$ very small since $\delta \approx 0$ and $\beta = 1 / (0.01)^2 = 10{,}000$.  We use $K = 20$ probe vectors, though we have found that even smaller values for $K$ do not change the results.  We employ $U = 400$ maximum CG steps with early termination if the residual error drops below the threshold $\epsilon_\text{max} = 10^{-4}$.  Results are displayed in Fig. \ref{acc}. 

\subsubsection{Performance vs. Sparsity} In Fig. \ref{acc}(a), we consider the dense Gaussian dictionary with $N = \lfloor D / 4 \rfloor$, which is typical in a \emph{compressed sensing} setting.
Let the \emph{sparsity factor} $f \in [0, 1]$ determine the number of non-zero coefficients in the latent signal $\bd z^*$ as $d = \lfloor f \cdot D\rfloor$. We vary $f$ and observe its impact on NRMSE.  At low $f$, all algorithms perform well except VI.  As $f$ increases, EM/IRLS, AMP, and CoFEM exhibit the same decay in performance, while Seq decays more rapidly.  In Fig. \ref{acc}(b), we consider the DCT dictionary with $N = \lfloor D / 3 \rfloor$ measurements, again varying $f$.  We see the same overall trend as in Fig. \ref{acc}(a).
In Fig. \ref{acc}(c), we compare NRMSE versus sparsity level $f$ for the convolutional dictionary.  The decay rate is set to $\rho = 0.04$.  EM/IRLS, Seq, and CoFEM have near-perfect NRMSE at all $f$. However, VI fails again due to its biased objective and AMP fails due to the convolutional dictionary not being zero-mean and sub-Gaussian.

\subsubsection{Performance vs. Undersampling} In Fig. \ref{acc}(d), we revisit the Gaussian dictionary and fix $f = 0.06$.  We vary the \emph{undersampling rate} $r > 0$, where $N = \lfloor D / r \rfloor$.  We see that EM/IRLS, AMP, and CoFEM have similar performance at all $r$, while VI and Seq degrade more rapidly with increasing $r$.  A similar trend is shown in Fig. \ref{acc}(e), in which we vary $r$ for the DCT dictionary while fixing $f = 0.12$.  

\subsubsection{Performance vs. Decay Rate} Finally, in Fig. \ref{acc}(f), we consider the convolutional dictionary with fixed $f = 0.2$ and vary the decay factor $\rho$.  Smaller $\rho$ leads to slower decay of the exponential filter, which increases the correlations between the columns of $\bd \Phi$.  We observe that EM/IRLS, Seq, and CoFEM are all robust to changes in $\rho$.  The performance of VI is heavily correlated with $\rho$, while AMP diverges for all values of $\rho$.   

CoFEM is the only algorithm that is as robust as EM/IRLS to changes in sparsity level, undersampling rate, and correlation between dictionary columns.  We attribute this robustness to the fact that CoFEM performs an unbiased estimation of the posterior variances, regardless of the dictionary structure.

\begin{figure*}
\begin{center}
\includegraphics[scale=0.47]{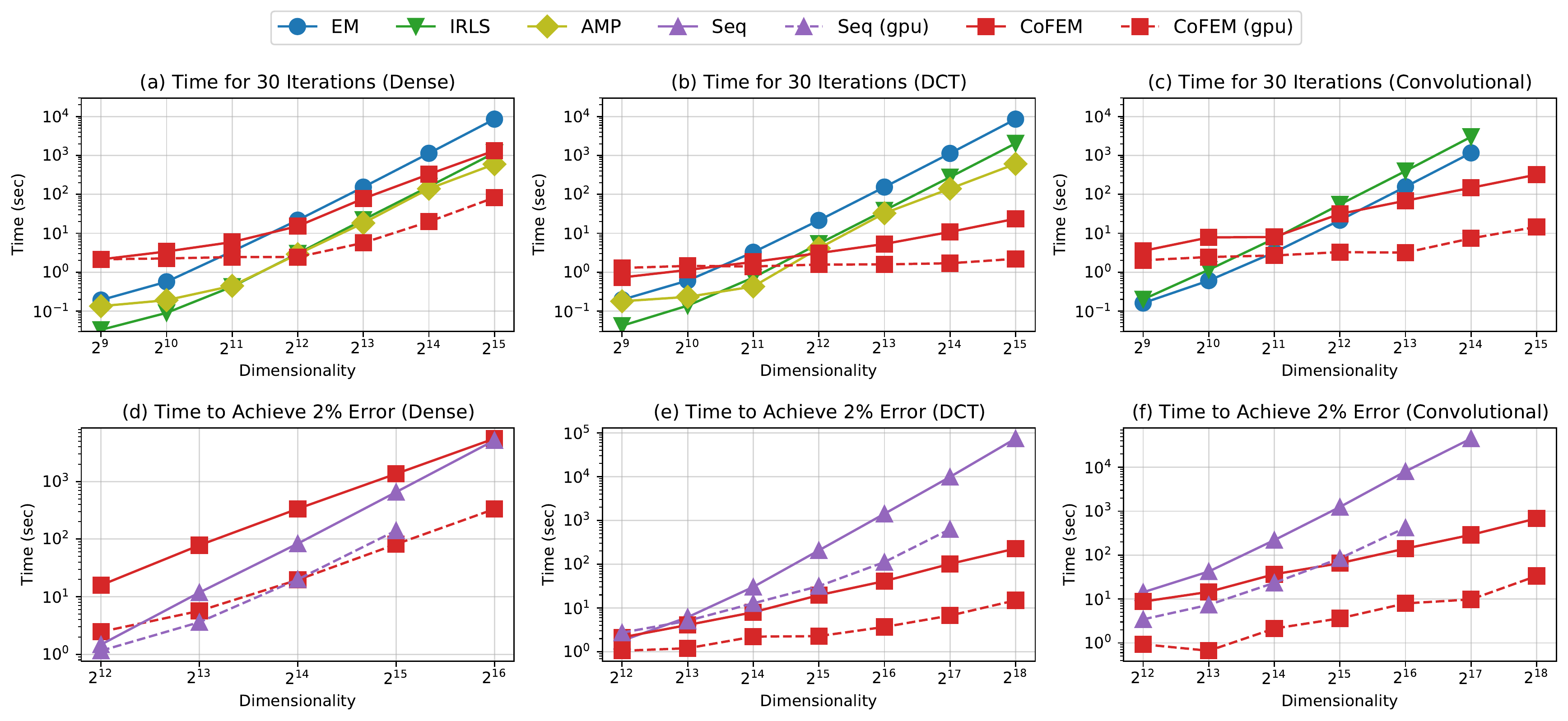} 
\caption{Comparing the scalability of different SBL inference schemes on log-log scales.  Some lines omit points for high dimensions due to memory issues.} \label{scale}
\end{center}
\end{figure*}

\subsection{Scalability Analysis}

We now analyze the scalability of CoFEM and other algorithms, by considering how computation time and memory requirements change as $D$ is increased.  In all settings, CoFEM has $K = 20$ probes and $U = 400$ maximum CG steps (with $\epsilon_\text{max} = 10^{-4}$), regardless of $D$.  The preconditioner employs $\theta_j = 1$ for all $j$.  Results are in Fig. \ref{scale}.

\subsubsection{CoFEM vs. EM-based Algorithms} We begin by comparing CoFEM against the other EM-based algorithms (EM, IRLS, AMP).  In each setting, we run each EM-based algorithm for $T= 30$ EM iterations, as this is sufficient to obtain low NRMSE (i.e. $<2\%$).    VI is omitted because it is unable to reach this threshold in many cases, as illustrated by Fig. \ref{acc}.  

For our first setting (Fig. \ref{scale}(a)), we consider the dense Gaussian $\bd \Phi$ with $D = 2^p$ for $p = \{9, 10, \ldots, 15\}$.  For each $D$, we let $N = \lfloor D / 4 \rfloor$ and $d = \lfloor 0.06 D \rfloor$.  As $D$ increases, we observe that CoFEM becomes faster than EM by an order of magnitude.  CoFEM is slower than AMP and IRLS, yet the gap decreases for large $D$.  This is because it takes $\tau_D = O(DN)$-time to apply a dense $\bd \Phi$ to a vector, so CoFEM has the same asymptotic complexity as AMP for the dense case (see Table \ref{tab:complex}).  Furthermore, due to its space-saving and parallelization-friendly design, CoFEM can exploit a GPU\footnote{We use a 16-GB Nvidia T4 GPU and a 32-GB, 2.3 GHz Intel Xeon CPU.} to be up to $7\times$ faster than IRLS/AMP and $100\times$ faster than EM.

Next, in Fig. \ref{scale}(b) and Fig. \ref{scale}(c), we repeat the experiment of increasing $D$ for the two structured $\bd \Phi$ (DCT and convolutional).  In both cases, there are fast algorithms for applying $\bd \Phi$ to an arbitrary vector in $\tau_D = O(D \log D)$-time.  As a result, CoFEM is faster at high $D$ than all other algorithms.  In the DCT case, we let $d = \lfloor 0.12 D \rfloor $ and $N = \lfloor D / 3 \rfloor$ for all $D$.  We observe that CoFEM can be faster than EM by up to $360\times$ on the CPU and $3800\times$ on the GPU, reducing over two hours of computation for EM at $D = 2^{15}$ to two seconds.  In the convolutional case, we let $d = \lfloor0.2 D\rfloor$.  For $D = 2^{15}$, EM and IRLS have memory issues due to the large $N = D$.  AMP is unable to yield a sensible solution (see Fig. \ref{acc}(c)).  In contrast, CoFEM (on both CPU and GPU) is accurate while being faster than EM/IRLS and not experiencing memory issues at high $D$.

\subsubsection{CoFEM vs. Seq} Finally, we compare CoFEM to the sequential algorithm.  Both of these algorithms have low space complexity (see Table \ref{tab:complex}), 
preventing memory issues at very high $D$.  Since CoFEM and Seq have different optimization procedures, we run both algorithms until they reach 2\% NRMSE.  We repeat the experimental settings from Fig. \ref{scale}(a)-(c) at higher dimensions $D = 2^{p}$ for $p =  \{12, 13, \ldots, 18\}$. 
 
 For the dense dictionary (Fig. \ref{scale}(d))\footnote{We could not run $D > 2^{16}$ for the dense case because there is not enough memory on our devices to store all the entries of $\bd \Phi$.}, we observe that Seq is faster than CoFEM on the CPU for low $D$, but the gap decreases for high $D$.  On the GPU, CoFEM is faster than Seq at moderate $D$ due to its greater ability to exploit parallelized hardware.  For the two structured dictionaries (Fig. \ref{scale}(e) and \ref{scale}(f)), CoFEM is faster across all settings.  We further observe that for many of the larger dimensions, Seq inevitably encounters memory issues as a covariance-based method.
 In contrast, CoFEM has no such issue due to its low space complexity.  For example, at $D = 2^{18}$, CoFEM can leverage the GPU to be $5000\times$ faster than Seq.
 In summary, CoFEM's ability to obviate covariance computation provides many advantages in terms of scalability over existing SBL inference schemes.
 
We emphasize that across all of our experiments, $U$ and $K$ are kept at small values despite substantially increasing $D$, demonstrating CoFEM's scalability at very high dimensions.
 
\begin{figure}
\begin{center}
\includegraphics[scale=0.45]{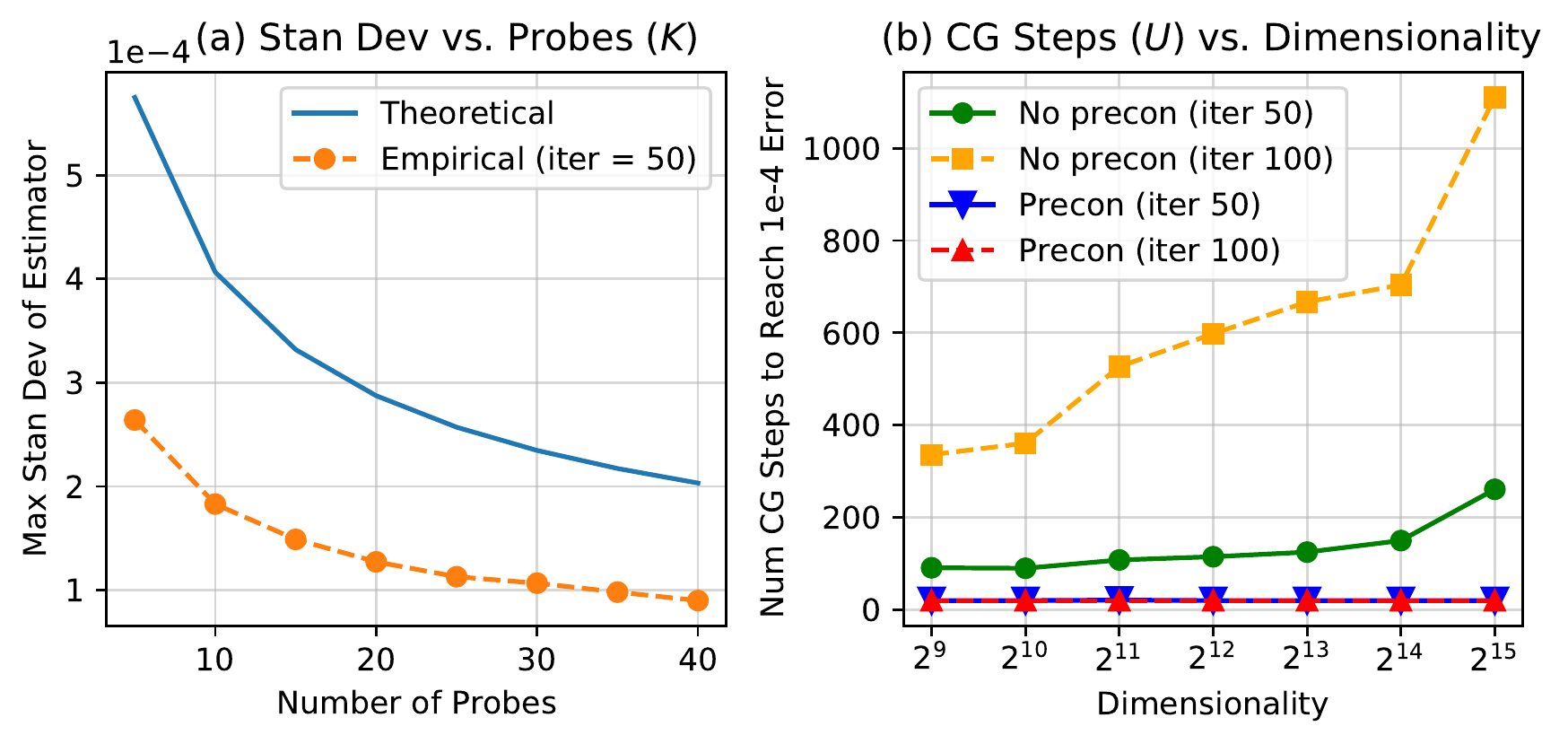}
\end{center}
\caption{Empirical insights into $U$ and $K$ based on theory of CoFEM.} \label{fig:emp}
\end{figure}

\begin{figure*}[h!]
\begin{center}
\includegraphics[scale=0.5]{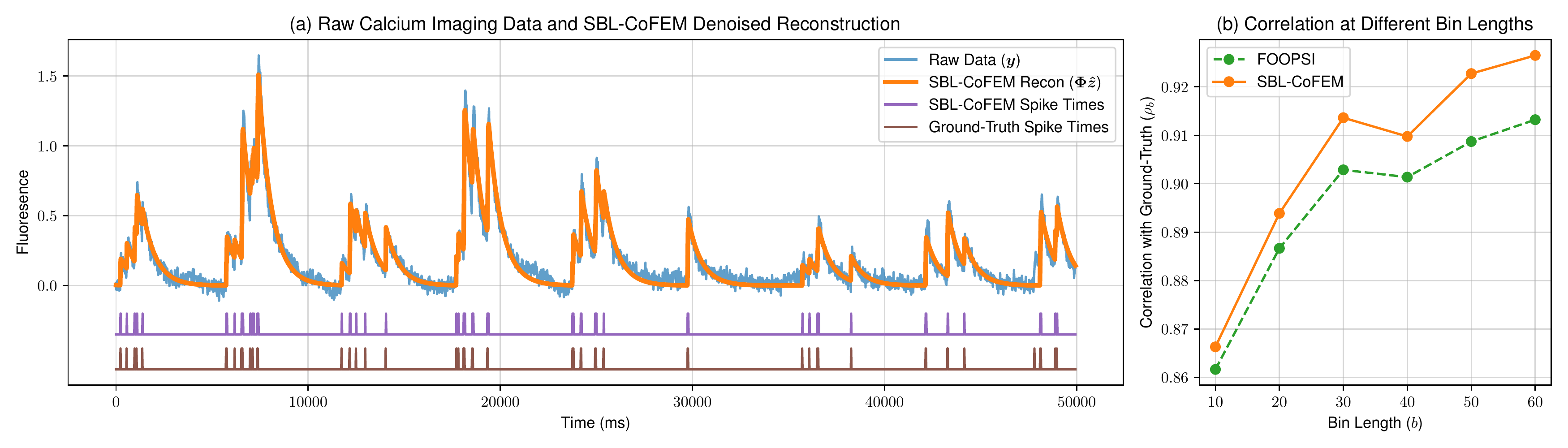}
\end{center}
\caption{Results of running SBL-CoFEM for calcium imaging and comparison with FOOPSI.}  \label{calc-res}
\end{figure*}

 \subsection{Empirical Insights from the Theory of CoFEM}
 Finally, we give some insights into our theory for CoFEM's hyperparameters $U$ and $K$ (Section \ref{sec:theory}).  We use the DCT dictionary setup of Section \ref{setup} with $N = \lfloor D / 3 \rfloor$ measurements.  
 
 In Fig. \ref{fig:emp}(a), we let $D = 1024$ and plot the relationship between (1) the number of probes $K$ and (2) the maximum standard deviation $\max_{j=1}^D \idx[j] \nu T$ over all $D$ coordinates of the diagonal estimator at iteration $T = 50$ of CoFEM.  We present a ``theoretical" curve calculated using Theorem \ref{var-bound-thm} with $\bd \Theta = \bold I$ and an ``empirical" curve in which each $\idx[j] \nu T$ is estimated through the empirical standard deviation of $1000$ Monte Carlo trials.  Fig. \ref{fig:emp}(a) shows that our theoretical bound (1) captures the true decay of the standard deviation as a function of $K$ and (2) accurately upperbounds the empirical curve after $T$ \emph{finite} iterations despite Theorem \ref{var-bound-thm}'s condition of $T \to \infty$.

 In Fig. 3(b), we consider different orders of magnitude for $D$ and plot the number of CG steps $U$ needed to achieve small error $\epsilon_{\text{max}} = 10^{-4}$ at the $T = 50$-th and $T = 100$-th iterations of CoFEM. We explore the impact of our diagonal preconditioner (Section \ref{sec:preconditioning}) on $U$.  The figure reveals that our theoretical insights from Section \ref{sec:theory-u} and Theorem \ref{kappa-bound} hold in practice: (1)  when there is no preconditioner and $\epsilon_{\text{max}}$ is fixed, $U$ needs to grow substantially with increasing iterations $T$ and/or dimensionality $D$, yet (2) when a preconditioner is used, $U$ can be small and constant for large $T$ and large $D$.

\section{Real-Data Experiments} \label{sec:experiments}
We now demonstrate the utility of CoFEM for {two real data} applications -- calcium deconvolution and MRI reconstruction.

\subsection{Calcium Deconvolution}

Calcium imaging is a widely used tool in neuroscience for monitoring the electrical activity of neurons {\cite{grienberger2012imaging}}.  It is a method for indirectly observing the spiking activity of a neuron through a fluorescence trace $\bd y \in \R^D$, approximated as the convolution of an intrinsically sparse spiking pattern $\bd z^* \in \R^D$ with a decaying calcium response $\bd \phi \in \R^D$.  In calcium deconvolution, we aim to recover $\bd z^*$ from $\bd y$ and $\bd \phi$.

\subsubsection{SBL Model} \label{sec:calc-model}
We can cast calcium deconvolution as a SBL problem with a dictionary $\bd \Phi \in \R^{D \times D}$ consisting of delayed (and truncated) versions of $\bd \phi$ as its columns, similar to the setting in Section \ref{sec:dict-types}. The data is then assumed to be generated as $\bd y = \bd \Phi \bd z^* + \bd \varepsilon$, where each $\varepsilon_i \sim \mathcal{N}(0, 1/\beta)$.
Since $\bd z^*$ for calcium deconvolution is a non-negative vector, we employ SBL with \emph{non-negativity constraints} (Section \ref{sec:sbl-nonneg-cofem}).   

\subsubsection{Spike Inference}  \label{sec:calc-inference}
We use the CoFEM inference algorithm adapted for non-negativity constraints.  Since $\bd \Phi$ represents a discrete-time convolution, we can efficiently apply $\bd \Phi$ to any vector $\bd v$ through fast Fourier transforms and an element-wise product.
Non-negative SBL yields a rectified Gaussian posterior $p(\bd z \given \bd y, \bhat \alpha)$ over the latent spikes $\bd z$.  To obtain a point estimate $\bhat z$, we find a \emph{filtered mode}.
Specifically, we first filter $\bd z$ by selecting components $z_j$ that are highly likely to be non-zero, i.e. $z_j$ such that $p(z_j = 0 \given \bd y, \bhat \alpha) < q$, where $q$ is some small percentile (e.g. $0.05$, $0.01$). This query is possible only because SBL models uncertainty in $\bd z$.  Setting the unselected components of $\bd z$ to zero, we then find the most likely values for all selected $z_j$'s according to the posterior, resulting in $\bhat z$. 
More details can be found in Appendix \ref{appendix-mode}.  This is analogous to thresholding heuristics commonly used by  $\ell_1$-based sparse coding algorithms \cite{friedrich2017fast}.  However, unlike those value-based strategies, the percentile filtering for SBL is value-agnostic and instead operates on the learned posterior.

\subsubsection{Data and Hyperparameters}
We apply SBL to five fluoresence traces from the GENIE dataset \cite{akerboom2012optimization}.  
Each $\bd y$ contains data at a sampling rate of $\nu = 60$ Hz for a total of $D = 14{,}400$ time points, which is a high-dimensional problem.  The response $\bd \phi$ has $\phi_i = (1 -\psi)^{i-1}$ for $\psi = 1 / (\nu \times 0.7) = 0.0238$, a widely-used value for the calcium indicator GCaMP6f.  CoFEM employs $T = 20$, $K = 20$ and $U = 400$. The noise precision $\beta$ is estimated from $\bd y$ through a Fourier domain procedure, as described in \cite{pnevmatikakis2016simultaneous}.  To obtain $\bhat z$, we use a filter percentile of $q = 0.05$.

\subsubsection{Results}                
To evaluate $\bhat z$, we employ the following standard practice \cite{pnevmatikakis2016simultaneous}: The GENIE dataset has ground-truth times for neural spikes.  Let $\bd z^* \in \R^D$ be a zero-one vector indicating when true spiking occurred.  For bin length $b$, we reduce $\bd z^*$ and $\bhat z$ to vectors $\bd c^*$ and $\bhat c$ of length $\lceil D / b \rceil$ by summing across windows of $b$ consecutive components.  We then compute the Pearson correlation coefficient $\rho_b$ between $\bd c^*$ and $\bhat c$.  A high value for $\rho_b$ indicates agreement between $\bhat z$ and $\bd z^*$.   

Fig. \ref{calc-res}(a) plots sample SBL-CoFEM outputs and compares its inferred spike times with the ground truth.  Fig. \ref{calc-res}(b) shows an averaged curve over the five traces of $\rho_b$ versus $b$ at various bin lengths $b \in \{10, 20, 30, 40, 50, 60\}$ for SBL-CoFEM and a {popular} $\ell_1$-based method called FOOPSI  \cite{vogelstein2010fast}. The figure shows that SBL-CoFEM outperforms FOOPSI, with the gap growing for larger bin sizes $b$.  



\subsection{Multi-Contrast MRI Reconstruction}
Magnetic resonance imaging (MRI) is one of the dominant modalities for imaging the human body \cite{nishimura2010principles}.  The standard data acquisition practice samples a set of points (called ``$k$-space") $\bd k \in \mathbb{C}^N$ from the two-dimensional Fourier transform (2DFT) of the image $\bd x \in \mathbb{R}^D$.\footnote{Though MRIs are generally complex-valued, the data we use here are real-valued.  See \cite{bilgic2011multi} for how to generalize SBL to the complex case.}  In practice, one may aim to collect $N < D$ points to reduce the amount of time a patient needs to remain in the scanner.  However, doing so leads to an ill-posed inverse problem $\bold M \bold F \bd x = \bd k$ for $\bd x$, where $\bold F \in \mathbb{C}^{D \times D}$ is the 2DFT and $\bold M \in \R^{N \times D}$ is an undersampling operator.  Thus, compressed sensing strategies often exploit the sparsity of $\bd x$ with respect to some transform for accurate reconstruction. 


In \emph{multi-contrast} MRI reconstruction, there are $L$ images $\bd x_1, \ldots, \bd x_L$ of an object that one wishes to recover from corresponding undersampled $k$-space measurements $\bd k_1, \ldots, \bd k_L$.  Bilgic et al. \cite{bilgic2011multi} showed that SBL with \emph{multi-task learning} (Section \ref{multi-task-cofem}) achieves successful joint recovery of the multiple images.  They outperformed $\ell_1$-based methods by exploiting common sparsity patterns among the horizontal/vertical image gradients {(i.e. row-wise/column-wise finite differences)} of $\bd x_1, \bd x_2, \ldots, \bd x_L$.  {However, the main drawback is the computation time for reconstruction, which is many times slower than $\ell_1$ methods.} We demonstrate how  CoFEM can accelerate this method while maintaining its superior performance.     

\begin{figure}
\begin{center}
\includegraphics[scale=0.45]{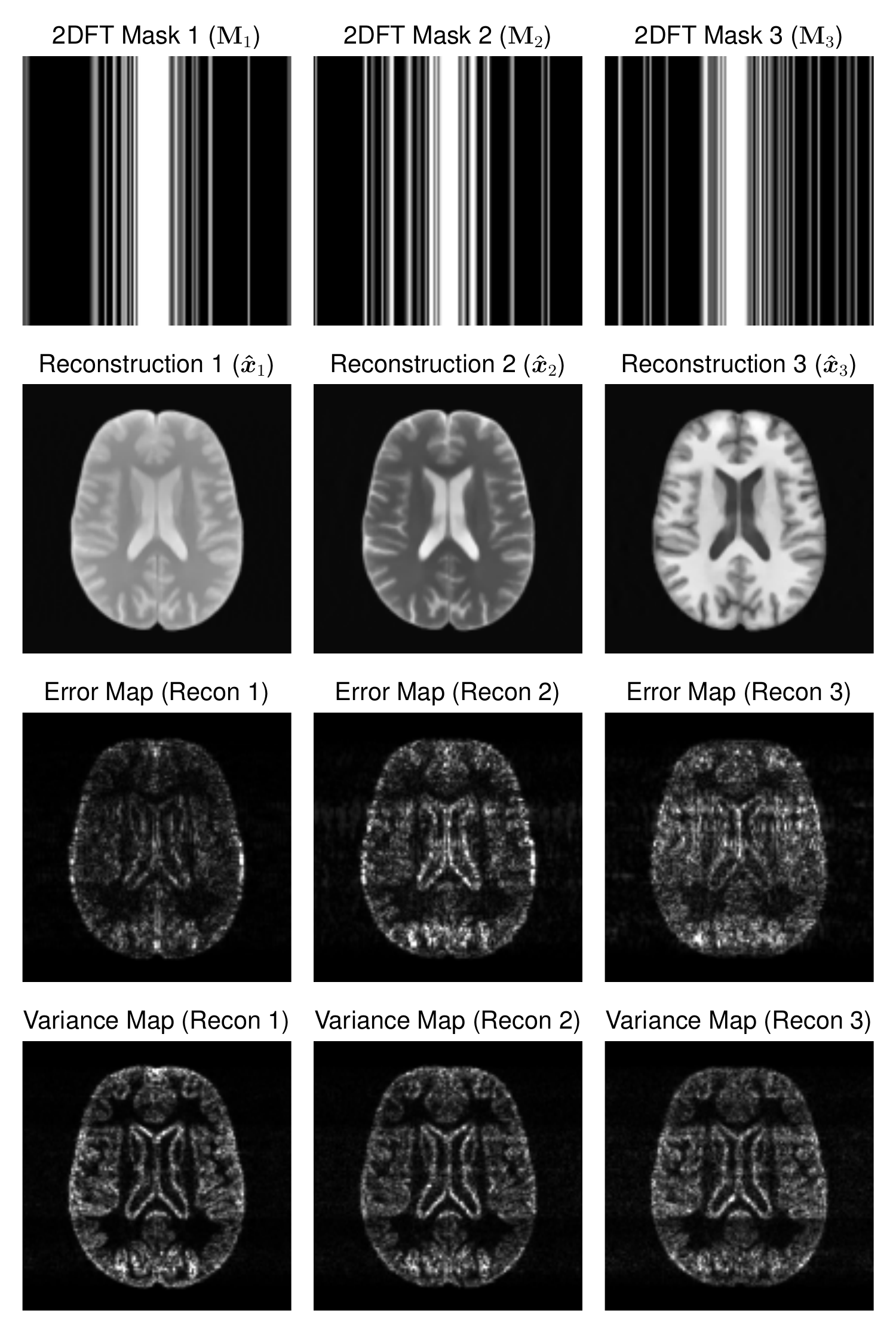}
\end{center}
\caption{Undersampling $k$-space masks and SBL-CoFEM reconstrutions for the SRI24 atlas.  Error maps are scaled by 15$\times$ to aid visualization.} \label{mri-fig}
\end{figure}

\subsubsection{SBL Model and Inference} \label{mri-model}
For each contrast $\ell$, let $\bd \Phi_\ell = \bold M_\ell \bold F$ denote the $\ell$-th undersampled 2DFT operator with undersampling mask $\bold M_\ell$. The task is to infer the sparse latent vector $\bd z^\text{horz}_\ell = \partial^\text{horz} \bd x_\ell\in \R^D$, where $\partial^\text{horz}$ denotes the horizontal image gradient operator.  Note that $\partial^\text{horz}$ is a convolution, implying that it corresponds to a diagonal matrix $\bd \Delta^\text{horz}$ in the Fourier domain.
Let the observed data be $\bd y^\text{horz}_\ell = \bold M_\ell \bd \Delta^\text{horz} \bold M_\ell^\top \bd k_\ell$ for all $\ell$.  We impose a multi-task SBL model on $(\bd z^\text{horz}_\ell, \bd y^\text{horz}_\ell, \bd \Phi_\ell)$, as per Eq. \eqref{multi-model}.  We use a shared $\bd \alpha^\text{horz}$ to ensure that we will learn grouped sparsity patterns among $\{\bd z_\ell^\text{horz}\}_{\ell=1}^L$. The same procedure is repeated for $\bd z_\ell^\text{vert}$ based on the operator $\partial^\text{vert}$.
Further details are in \cite{bilgic2011multi}.

We employ CoFEM for multi-task SBL (Section \ref{multi-task-cofem}) to recover $p(\bd z_\ell^\text{horz} \given \bd y^\text{horz}_\ell, \bhat \alpha^\text{horz})$ and $p(\bd z_\ell^\text{vert} \given \bd y_\ell^\text{vert}, \bhat \alpha^\text{vert})$ for all $\ell$. The reconstructed image $\hat{\bd x}_\ell$ are computed from these posterior distributions. More details are given in Appendix \ref{appendix-final-recon}.

\begin{table}
\centering
\caption{Results on Multi-Contrast MRI Reconstruction} \label{mri-table}
\begin{tabular}{lllll}
\toprule
\textbf{Algorithm} & \textbf{NRMSE} & \textbf{Computation Time} \\
\hline
SparseMRI & 5.4\% & 22.3 min  \\
SBL-Seq & 3.4\% & 89.9 min\\
SBL-CoFEM  & 2.9\% & 2.5 min (CPU), 0.2 min (GPU) \\
\bottomrule
\end{tabular}
\end{table}

\subsubsection{Data and Hyperparameters}
We consider the SRI24 atlas \cite{rohlfing2010sri24}, a set of $L = 3$ MRI contrasts with dimensions $200 \times 200$ for a total of $D = 40{,}000$ pixels.  For each image $\bd x_\ell^*$, we undersample its 2DFT by a factor of four in the horizontal dimension, observing $N = 10{,}000$ points to form $\bd k_\ell$.  The mask $\bold M_\ell$ is randomly determined according to a power rule favoring the center of $k$-space \cite{lustig2007sparse}. For CoFEM, we have $T = 15$, $\beta = 10^{6}$, $K = 8$, $\epsilon_\text{max} = 10^{-5}$ and $U = 200$.  The algorithm is not sensitive to variations in these values.               

\subsubsection{Results}
Fig. \ref{mri-fig} provides images of the masks $\bold M_\ell$ and SBL-CoFEM's reconstructions $\bhat x_\ell$. Success is measured through NRMSE between the vectorized forms of $\bhat x \in \R^{D \cdot L}$ and $\bd x^* \in \R^{D \cdot L}$.
Table \ref{mri-table} compares SBL-CoFEM against SparseMRI ($\ell_1$-based compressed sensing \cite{lustig2007sparse}) and SBL-Seq (SBL with sequential algorithm~\cite{bilgic2011multi}).  Although SBL-Seq has lower NRMSE than SparseMRI, it requires high computation time.  In contrast, SBL-CoFEM attains the lowest error and can exploit the GPU to be 450$\times$ faster than SBL-Seq.

{Finally, the bottom half of Fig. \ref{mri-fig} displays error maps of absolute differences between $\bhat x_\ell$ and $\bd x^*_\ell$, along with \emph{variance maps} for each image.  Each variance map captures the model's confidence over different areas of its reconstruction; pixels with high variance indicate more potential to deviate from the point estimate $\bhat x_\ell$. SBL can create variance maps because it models uncertainty; non-Bayesian methods that do not model uncertainty (e.g. $\ell_1$ methods) cannot generate these maps. 
Appendix \ref{var-maps} explains how SBL-CoFEM can generate these variance maps using the diagonal estimation rule.   
    The variance maps bear similarity to the ground-truth error maps, suggesting that SBL-CoFEM can predict its errors in reconstruction.  
}

\section{Conclusion}
We developed covariance-free expectation-maximization (CoFEM) to accelerate sparse Bayesian learning (SBL).  By solving linear systems to circumvent matrix inversion, CoFEM exhibits superior time-efficiency and space-efficiency over existing SBL inference schemes, especially when the dictionary $\bd \Phi$ admits fast matrix-vector multiplication.  We theoretically analyzed CoFEM's hyperparameters, such as the number of linear systems and number of solver steps, showing that they can remain small even at high dimensions.  Coupled with GPU acceleration, CoFEM can be up to thousands of times faster than other SBL methods without sacrificing sparse coding accuracy.  Finally, we used CoFEM for real-data applications, showing that it can adapt to multi-task learning and non-negativity constraints, while enabling SBL to be competitive with non-Bayesian methods in accuracy and scalability.



%

\appendices

\section{Proof of Theorem \ref{var-bound-thm}} \label{var-bound-proof}

\begin{proof}
We begin by characterizing $\bhat \Sigma := \lim_{t \to \infty} \bidx \Sigma t$.  By the block matrix inversion formula \cite{strang2006linear}, any symmetric positive definite matrix with diagonal blocks $\bold X, \bold Z$ and off-diagonal block $\bold Y$ and can be inverted as 
\begin{align}
\begin{bmatrix}\bold X &  \hspace{-0.5em} \bold Y  \\ \bold Y^\top &  \hspace{-0.5em} \bold Z \end{bmatrix}^{-1} \hspace{-1em} = 
\begin{bmatrix}
\bold X^{-1} + \bold X^{-1} \bold Y \bold W \bold Y^\top \bold X^{-1} &  \hspace{-0.5em}-\bold X^{-1} \bold Y \bold W \\
-\bold W  \bold Y^\top \bold X^{-1} &  \hspace{-0.5em}\bold W
\end{bmatrix} \hspace{-0.2em} , \hspace{-0.2em} \label{block-matrix}
\end{align}
where $\bold W := (\bold Z - \bold Y^\top \bold X^{-1} \bold Y)^{-1}$ is the inverse Schur complement.  We apply Eq. \eqref{block-matrix} to $\bidx \Sigma t$, with $\bold X:= \beta \bold \Phi^\top_\mathcal{S} \bd \Phi_\mathcal{S} + \mathrm{diag}\{\bidx[\mathcal{S}] \alpha t\}$, $\bold Y := \beta \bd \Phi^\top_\mathcal{S} \bd \Phi_\mathcal{U}$, and $\bold Z :=  \beta \bold \Phi^\top_\mathcal{U} \bd \Phi_\mathcal{U} + \mathrm{diag}\{\bidx[\mathcal{U}] \alpha t\}$.  As $t \to \infty$, we have $\idx[j] \alpha t \to \infty$ for $j \in \mathcal{U}$, which forces $\bold W \to \bd 0$, where $\bold 0$ is the zero-matrix.  Then, by Eq. \eqref{block-matrix},
\begin{align}
\bhat \Sigma = \lim_{t \to \infty} \bidx \Sigma t = \begin{bmatrix}
 (\beta \bold \Phi^\top_\mathcal{S} \bd \Phi_\mathcal{S} + \mathrm{diag}\{\bhat \alpha_\mathcal{S}\})^{-1} & \bold 0 \\
\bold 0 & \bold 0
\end{bmatrix}. \label{sigma-converge}
\end{align}  

Eq. \eqref{inactive-var} follows from applying Eq. \eqref{sigma-converge} to Lemma \ref{diag-var-lemma} for $j \in \mathcal{U}$; since all rows of $\bhat \Sigma$ that correspond to $\mathcal{U}$ are zero, the estimator's standard deviation must also converge to zero.  

We now prove Eq. \eqref{active-var} for an active index $j \in \mathcal{S}$.  Let $\psi_j > 0$ be any positive real number.  Using Lemma \ref{diag-var-lemma} and Eq. \eqref{sigma-converge}, we can bound $\hat{\nu}_j := \lim_{t \to \infty} \idx[j] \nu t$ with 
\begin{align}
\hspace{-0.5em} \hat{\nu}_j  = \sqrt{\frac{1}{K} \hspace{-0.3em}  \sum_{\substack{j' \in \mathcal{T}_j}} \hspace{-0.3em}  \mathrm {\hat{\Sigma}}_{j, j'}^2} \leq \frac{1}{\sqrt{K}} \sqrt{({\hat{\Sigma}}_{j, j} - \psi_j)^2 + \hspace{-0.5em} \sum_{\substack{j' \in \mathcal{T}_j}} \mathrm {\hat{\Sigma}}_{j, j'}^2},  \label{diag-bound}
\end{align}
where $\mathcal{T}_j := \mathcal{S} \setminus \{j\}$.
Let $\bold \Psi \in \R^{|\mathcal{S}| \times |\mathcal{S}|}$ be a diagonal matrix with $\psi_j$ for all $j \in \mathcal{S}$ along its diagonal.  Let $\bd e_j$ be the standard unit vector in $\R^{|\mathcal{S}|}$ corresponding to $j$.  From Eq. \eqref{diag-bound}, we have
\begin{align}
\hspace{-0.6em} \hat{\nu}_j &\leq  \tfrac{1}{\sqrt{K}}  \norm{(\bhat \Sigma_{\mathcal{S}, \mathcal{S}} - \bold \Psi) \bd e_j}_2 \nonumber \\&\leq  \tfrac{1}{\sqrt{K}}  \norm{\bhat \Sigma_{\mathcal{S}, \mathcal{S}} - \bold \Psi}_2 =  \tfrac{1}{\sqrt{K}}\norm{\bold \Psi(\bhat \Sigma_{\mathcal{S}, \mathcal{S}}^{-1} - \bold \Psi^{-1})\bhat \Sigma_{\mathcal{S}, \mathcal{S}}}_2 \nonumber \\
&\leq  \tfrac{1}{\sqrt{K}}  \norm{\bold \Psi (\bhat \Sigma_{\mathcal{S}, \mathcal{S}}^{-1} - \bold \Psi^{-1})}_2  \norm{\bhat \Sigma_{\mathcal{S}, \mathcal{S}}}_2 \nonumber \\
&=  \tfrac{1}{\sqrt{K}}  \norm{\bold \Psi (\beta \bd \Phi^\top_\mathcal{S} \bd \Phi_\mathcal{S} + \mathrm{diag}\{\bhat \alpha_\mathcal{S}\} - \bold \Psi^{-1})}_2 \norm{\bhat \Sigma_{\mathcal{S}, \mathcal{S}}}_2,  \label{bound-prod}
\end{align}  
where the last step uses Eq. \eqref{sigma-converge} to expand $\bhat \Sigma_{\mathcal{S}, \mathcal{S}}^{-1}$.  

We now perform a change-of-variables: Let $\bold \Theta \in \R^{|\mathcal{S}| \times |\mathcal{S}|}$ be any diagonal matrix of positive diagonal values.  We define $\bd \Psi := (\beta \bd \Theta + \mathrm{diag} \{\bhat \alpha_\mathcal{S} \})^{-1}$ and re-write Eq. \eqref{bound-prod} as     
\begin{align}
\hat{\nu}_j  &\leq \tfrac{1}{\sqrt{K}} \norm{(\beta \bd \Theta + \mathrm{diag} \{\bhat \alpha_\mathcal{S} \})^{-1} (\beta \bd \Phi^\top_\mathcal{S} \bd \Phi_\mathcal{S} - \beta \bd \Theta)}_2  \norm{\bhat \Sigma_{\mathcal{S}, \mathcal{S}}}_2 \nonumber \\
&\leq  \tfrac{1}{\sqrt{K}}  \norm{(\beta \bd \Theta)^{-1} (\beta \bd \Phi^\top_\mathcal{S} \bd \Phi_\mathcal{S} - \beta \bd \Theta)}_2  \norm{\bhat \Sigma_{\mathcal{S}, \mathcal{S}}}_2 \nonumber \\
&=  \tfrac{1}{\sqrt{K}}   \norm{ \bd \Theta^{-1} \bd \Phi^\top_\mathcal{S} \bd \Phi_\mathcal{S} - \bold I}_2 \norm{\bhat \Sigma_{\mathcal{S}, \mathcal{S}}}_2. \label{final-bound}
\end{align} 

Finally, we can bound the last term of Eq. \eqref{final-bound} with
\begin{align}
\norm{\bhat \Sigma_{\mathcal{S}, \mathcal{S}}}_2 = \frac{1}{\lambda_\text{min}(\bhat \Sigma_{\mathcal{S}, \mathcal{S}}^{-1})} \leq \frac{1}{\beta \cdot \lambda_\text{min}(\bd \Phi_\mathcal{S}^\top \bd \Phi_\mathcal{S})},
\end{align}
where we use the fact that the singular values coincide with the eigenvalues for a symmetric positive definite matrix.  
\end{proof}

\section{Proof of Theorem \ref{kappa-bound}} \label{kappa-bound-proof}

\begin{proof} The statement follows from Eq. \eqref{epsilon-bound} and showing that $\hat{\kappa} := \lim_{t \to \infty} \kappa(\boldidx {A'} t) \leq (1 + \xi) / (1 - \xi)$, where $\xi := \norm{\bd \Theta^{-1} \bd \Phi_\mathcal{S}^\top \bd \Phi_\mathcal{S} - \bold I}$.  Let $\bold \Delta := \beta \bd \Phi^\top \bd \Phi - \mathrm{diag}\{\beta \bd \theta\}$.  Then, 
\begin{align}
&\boldidx A t = \boldidx M t + \bold \Delta \implies \boldidx {A'} t = \bold I + (\boldidx M t)^{-\frac{1}{2}} \bold \Delta (\boldidx M t)^{-\frac{1}{2}} \nonumber \\
& \hspace{-0.7em} \implies \hat{\bold A}'  := \lim_{t \to \infty} \boldidx {A'} t = 
\bold I + \begin{bmatrix}
\hat{\bold M}_{\mathcal{S}, \mathcal{S}}^{-\frac{1}{2}} \bold \Delta_{\mathcal{S}, \mathcal{S}} \hat{\bold M}_{\mathcal{S}, \mathcal{S}}^{-\frac{1}{2}} & \bd 0 \\
\bd 0 & \bd 0
\end{bmatrix}, \label{aprime-mat}
\end{align} 
where $\hat{\bold M} := \mathrm{diag}\{\beta \bd \theta + \bhat \alpha\}$.  

Our goal is to bound $\hat{\kappa} = \lambda_\text{max}(\hat{\bold A}') / \lambda_\text{min}(\hat{\bold A}')$. Eq. \eqref{aprime-mat} shows that if $\eta$ is an eigenvalue of $\hat{\bold M}_{\mathcal{S}, \mathcal{S}}^{-1/2} \bold \Delta_{\mathcal{S}, \mathcal{S}} \hat{\bold M}_{\mathcal{S}, \mathcal{S}}^{-1/2}$, then $1 + \eta$ is an eigenvalue of $\hat{\bold A}'$.  This reduces our task to bounding $\eta$.  A matrix $\bold X \in \R^{D \times D}$ is similar to another matrix $\bold Y \in \R^{D \times D}$ if there exists an invertible matrix $\bold Z \in \R^{D \times D}$ such that $\bold Y = \bold Z^{-1} \bold X \bold Z$, and this further implies that $\bold X$ and $\bold Y$ have the same eigenvalues \cite{strang2006linear}.  Since $\bold X := \hat{\bold M}_{\mathcal{S}, \mathcal{S}}^{-1/2} \bold \Delta_{\mathcal{S}, \mathcal{S}} \hat{\bold M}_{\mathcal{S}, \mathcal{S}}^{-1/2}$ is similar to $\bold Y := \hat{\bold M}_{\mathcal{S}, \mathcal{S}}^{-1} \bold \Delta_{\mathcal{S}, \mathcal{S}}$ for $\bold Z := \hat{\bold M}_{\mathcal{S}, \mathcal{S}}^{1/2}$, it follows that $\eta$ is also an eigenvalue of $\bold Y$.  The absolute eigenvalues of a matrix cannot exceed its spectral norm, implying
\begin{align}
&\hspace{-0.5em} |\eta| \leq \norm{\hat{\bold M}_{\mathcal{S}, \mathcal{S}}^{-1} \bold \Delta_{\mathcal{S}, \mathcal{S}}}_2 = \norm{\hat{\bold M}_{\mathcal{S}, \mathcal{S}}^{-1} (\beta \bd \Phi_\mathcal{S}^\top \bd \Phi_\mathcal{S} -  \beta \bold \Theta)}_2\nonumber \\
&\hspace{-0.5em}  \leq \norm{(\beta \bd \Theta)^{-1} (\beta \bd \Phi_\mathcal{S}^\top \bd \Phi_\mathcal{S} -  \beta \bold \Theta)}_2 = \norm{\bd \Theta^{-1} \bd \Phi_\mathcal{S}^\top \bd \Phi_\mathcal{S} - \bold I}_2.
\end{align}  
It follows that $\lambda_\text{max}(\hat{\bold A}') \leq 1 + \norm{\bd \Theta^{-1} \bd \Phi_\mathcal{S}^\top \bd \Phi_\mathcal{S} - \bold I}_2$ and $\lambda_\text{min}(\hat{\bold A}') \geq 1 - \norm{\bd \Theta^{-1} \bd \Phi_\mathcal{S}^\top \bd \Phi_\mathcal{S} - \bold I}_2$, which bounds $\kappa(\hat{\bold A}')$.  
\end{proof}

\section{Details of SBL for Calcium Deconvolution}

%

\subsubsection{Filtered Mode} \label{appendix-mode}
Let $\mathcal{S} \subseteq \{1, 2, \ldots, D\}$ be the set of selected indices after percentile filtering of $p(\bd z \given \bd y, \bhat \alpha)$ recovered by CoFEM.  Let $\bd \Phi_\mathcal{S} \in \R^{N \times |\mathcal{S}|}$ denote the sub-matrix of $\bd \Phi$ composed of the columns corresponding to $\mathcal{S}$.  Then, the \emph{filtered mode} is the solution to the following problem: 
\begin{align}
\bhat u = \arg \min_{\bd u \geq \bd 0 \in \R^{|\mathcal{S}|}} ||\bd y - \bd \Phi_\mathcal{S} \bd u||_2^2 + \sum_{j \in \mathcal{S}} (\hat{\alpha}_j / \beta) u_j^2, \label{appendix-eq}
\end{align} 
which can be obtained using non-negative least-squares solvers.  Solving Eq. \eqref{appendix-eq} is fast in practice, because it is a low-dimensional problem (i.e. $|\mathcal{S}| \ll D$).  Our point estimate solution is $\bhat  z$, where $\hat{z}_j = u_j$ for $j \in \mathcal{S}$ and $\hat{z}_j = 0$ for $j \not \in \mathcal{S}$.

{
\section{Details of SBL for MRI Reconstruction} 

\subsubsection{Final Reconstruction}\label{appendix-final-recon}
Let $\bd \mu_\ell^\text{horz}$ and $\bd \mu_\ell^\text{vert}$ denote the respective means of $p(\bd z_\ell^\text{horz} \given \bd y^\text{horz}_\ell, \bhat \alpha^\text{horz})$ and $p(\bd z_\ell^\text{vert} \given \bd y_\ell^\text{vert}, \bhat \alpha^\text{vert})$. These quantities are combined through solving a constrained least-squares problem to yield a final reconstruction $\bhat x_\ell$: 
\begin{align}
\bhat x_\ell = \arg &\min_{\bd x_\ell} \norm{\partial^\text{horz} \bd x_\ell - \bd \mu_\ell^\text{horz}}_2^2 + \norm{\partial^\text{vert} \bd x_\ell - \bd \mu_\ell^\text{vert}}_2^2, \nonumber \\
&\text{ s.t. } \bold M_\ell \bold F \bd x_\ell = \bd k_\ell. \label{mri-least-squares}
\end{align}  
We use Parseval's Theorem \cite{oppenheim2001discrete} to cast Eq. \eqref{mri-least-squares} to the Fourier domain.  This converts $\partial^\text{horz}$ and $\partial^\text{vert}$ into diagonal matrices $\bold \Delta^\text{horz}$ and $\bold \Delta^\text{vert}$ in the Fourier domain, giving an element-wise separable problem with a closed-form solution \cite{bilgic2011multi}.  

\subsubsection{Variance Map} \label{var-maps}
For each MRI contrast $\ell$, SBL learns posterior distributions $\mathcal{N}(\bd \mu^\text{horz}_\ell, \bd \Sigma^\text{horz}_\ell)$ and $\mathcal{N}(\bd \mu^\text{vert}_\ell, \bd \Sigma^\text{vert}_\ell)$ for the image gradients. We use the fact that for $\bd z\sim\mathcal{N}(\bd \mu, \bd \Sigma)$ and a matrix $\bold E$, we have $\bold E \bd z\sim \mathcal{N}(\bold E \bd \mu, \bold E \bold \Sigma \bold E^\top)$. The solution to Eq. \eqref{mri-least-squares} has the form $\bhat x_\ell = \bold E_1 \bd \mu_\ell^\text{horz} + \bold E_2 \bd \mu_\ell^\text{horz} + \bd e$ for some matrices $\bold E_1, \bold E_2 \in \R^{D \times D}$ and some vector $\bd e \in \R^D$.  To obtain variance maps for $\bhat x_\ell$, we find the diagonal elements of its covariance matrix $\bd \Psi := \bold E_1 \bd \Sigma^\text{horz}_\ell \bold E_1^\top + \bold E_2 \bd \Sigma^\text{vert}_\ell \bold E_2^\top$.  We can do this by drawing probes and applying the diagonal estimation rule (Section \ref{simp-estep}).  In doing so, we need to apply $\bd \Sigma^\text{horz}_\ell = (\bd \Phi^\top \bd \Phi + \bhat \alpha^\text{horz})^{-1}$ and  $\bd \Sigma^\text{vert}_\ell = (\bd \Phi^\top \bd \Phi + \bhat \alpha^\text{vert})^{-1}$ to arbitrary vectors, which can be done via parallel CG (Section \ref{parallel-cg-alg}). 
}



\ifCLASSOPTIONcaptionsoff
  \newpage
\fi



%
\bibliographystyle{IEEEtran}
\bibliography{ref}
%
%
%
%




\end{document}